\documentclass[10pt,english,leqno]{siamltex1213}
\usepackage[T1]{fontenc}
\usepackage[latin9]{inputenc}
\usepackage{amsmath}
\usepackage{amssymb}
\usepackage{graphicx}

\makeatletter
\usepackage[left=3.0cm,right=3.0cm,top=2.2cm,bottom=2.2cm]{geometry}

\title{Calculation of a power price equilibrium\thanks{We thank the Oxford-Man Institute for providing historical prices used to calibrate our model, and ELEXON for providing historical data about the Balancing Mechanism used to determine physical characteristics of the power plants connected to the UK power grid. }} 

\author{Miha Troha\thanks{Mathematical Institute, Oxford University, Andrew Wiles Building, Radcliffe Observatory Quarter, Woodstock Road, Oxford OX2 6GG, United Kingdom, \email{troha@maths.ox.ac.uk}. This author was supported through grants from the Slovene human resources development and scholarship fund, and the Oxford-Man Institute.} \and Raphael Hauser\thanks{Mathematical Institute, Oxford University, Andrew Wiles Building, Radcliffe Observatory Quarter, Woodstock Road, Oxford OX2 6GG, United Kingdom, \email{hauser@maths.ox.ac.uk}. Associate Professor in Numerical Mathematics, and Tanaka Fellow in Applied Mathematics at Pembroke College, Oxford. This author was supported through grant EP/H02686X/1 from the Engineering and Physical Sciences Research Council of the UK.}}

\usepackage{tikz}
\usepackage{verbatim}
\usepackage{varwidth}
\usepackage{ragged2e}
\usetikzlibrary{shapes,arrows,positioning,fit,calc,backgrounds,decorations.shapes}

\tikzset{decorate sep/.style 2 args=
{decorate,decoration={shape backgrounds,shape=circle,shape size=#1,shape sep=#2}}}

\makeatother

\usepackage{babel}
\begin{document}
\maketitle
\begin{abstract}
In this paper we propose a tractable quadratic programming formulation
for calculating the equilibrium term structure of electricity prices.
We rely on a theoretical model described in \cite{troha2014theexistence},
but extend it so that it reflects actually traded electricity contracts,
transaction costs and liquidity considerations. Our numerical simulations
examine the properties of the term structure and its dependence on
various parameters of the model. The proposed quadratic programming
formulation is applied to calculate the equilibrium term structure
of electricity prices in the UK power grid consisting of a few hundred
power plants. The impact of ramp up and ramp down constraints are
also studied. 
\end{abstract}

\section{Introduction}

Since the deregulation of the electricity markets in the 1990s, the
modeling of electricity prices has attracted a lot of attention. A
first approach, nowadays called the non-structural approach, attempted
to model electricity prices by using the standard techniques that
had already been developed for modeling the prices of other financial
securities. The seminal paper \cite{lucia2000electricity} studied
the suitability of one and multi-factor Ornstein-Uhlenbeck processes
to model the spot price. As pointed out in this work, Gaussian distributions
alone can not be used for modeling spikes, and thus the Ornstein-Uhlenbeck
process was combined with a pure jump-process in \cite{hambly2009modelling}.
An extension to more general Levy processes has been described in
\cite{meyer-brandis2008multifactor}. However, as argued in \cite{benth2009theinformation},
such approaches cannot be directly used for pricing electricity derivatives,
because they neglect the non-storable nature of power. They all rely
on the non-arbitrage principle and thus implicitly assume that a buy-and-hold
strategy is possible. Such an assumption is clearly not realistic
for the electricity market. Thus, other models that try to capture
some of the physical properties of electricity have been proposed.

The seminal work \cite{barlow2002adiffusion} uses the supply and
demand stack to calculate the electricity prices. The model was extended
in \cite{howison2009stochastic} and \cite{carmona2013electricity},
where the supply stack was modeled as a function of the underlying
fuels (e.g. gas, coal, oil...) used to produce the electricity. A
closed form solution for electricity forward contracts, and for spark
and dark spread options was derived. These models capture some of
the physical properties of the power markets and in the literature
they are referred to as structural approaches. However, as argued
in \cite{robinson2005mathmodel}, models that include the ramp up/down
constraints of power plants and long term contracts are needed in
order to understand and prevent catastrophic events and market manipulation
as happened in California in 2001. 

This is the motivation behind a third, game theoretic approach, that
models the physical properties and decisions of market participants
more closely. The seminal work in this line of research was produced
by \cite{bessembinder2002equilibrium}, where a dependency between
a forward and a spot price in a two-stage market is studied. The market
consists of one producer and one consumer, who each aim to optimize
a mean-variance utility function. \cite{cavallo2005electricity} used
this work to study the impact of derivatives in the power market.
\cite{buhler2009valuation}, and \cite{buhler2009riskpremia} have
extended it to a multi-stage setting and derived a formula for a dynamic
equilibrium. The mean-variance objective function was replaced by
a general convex risk measure in \cite{demaeredaertrycke2012liquidity}.
\cite{troha2014theexistence} extended the work of \cite{buhler2009riskpremia}
to a setting with more than one producer and consumer, who optimize
their mean-variance objective functions. In contrast to other game
theoretic models, capacity and ramp up/down constraints of power plants
are included. By modeling the profit of power plants as a difference
between the power price and fuel costs together with emissions obligations,
this work also incorporates ideas from the structural approach. As
in \cite{clewlow1999amultifactor} and in \cite{clewlow1999valuing},
the model is consistent with observable fuel and emission prices.
However, \cite{troha2014theexistence} focuses only on the description
of the model and on the proof of the existence and uniqueness of the
solution. In this paper, we extend the model of \cite{troha2014theexistence}
to include transaction costs, liquidity constraints and actually traded
electricity contracts as well as propose a tractable quadratic programming
formulation to solve it numerically. The numerical results show the
dependence of the term structure of electricity prices on various
parameters of the model. Moreover, the algorithm is applied to a realistic
setting by incorporating the entire UK power grid, consisting of a
few hundred power plants.

From our model, it might not be immediately clear why market participants
do not execute all trades at the beginning of the planning horizon.
One reason for spreading out the trading activity is the availability
of contracts and their liquidity: contracts with a delivery far into
the future are much less liquid than, for example, day ahead contracts.
Liquidity is included in our model through increased costs of trading
(see Subsection \ref{sub:Costs-of-trading}). Another reason are delayed
cash flows: when trading is done through future contracts, some players
might be interested in entering into a position later and thus delay
the associated cash flows. In this paper we mainly postulate the use
of forward contracts. A simple extension to future contracts (assuming
a constant interest rate) is presented in Subsection \ref{sub:Future-contracts}.
A third reason for delayed trading is the exploitation of a trend-following
effect in the term structure: due to risk aversion of most players,
the term structure of forward or future contracts with fixed delivery
date is usually slightly upward sloping. A fourth reason are transaction
costs: when large trades, such as hedging of a whole power plant,
must be executed, they must be spread over a longer period of time
to decrease transaction costs and market impact. 

Our paper is organized as follows: To keep the paper self contained,
we introduce the main components of the term structure power price
model presented in \cite{troha2014theexistence} in Section \ref{sec:Problem-description}.
A quadratic programming formulation that can be used to calculate
the price equilibrium is described in Section \ref{sec:Algorithm}.
In Section \ref{sec:Extensions} we continue with various realistic
extensions of the model. In Section \ref{sec:Simulations}, we illustrate
the forecasting power of our model through numerical experiments.
In the first part of this section we study the term structure of electricity
prices on a simple example. In the second part we calculate the equilibrium
electricity price while modeling the entire system of UK power plants.
We conclude the paper in Section \ref{sec:Conclusions}.

\section{Organization of the UK Electricity Market\label{sec:Organization-of-market}}

Electricity supply and demand must match continuously in real-time,
which makes electricity grids very difficult to manage. At a high
level scale, the UK electricity market can be divided into two different
modes of operation: The first mode deals with electricity contracts
with more distant delivery periods (ranging from seasons to one hour).
The second mode is performed by the system operator (in the UK called
the National Grid) responsible for the micromanagement of the grid
in the last hour before the delivery of electricity. 

In the first mode electricity is traded through forward or future
contracts. Forward contracts start to be traded up to four years before
delivery. This period is usually referred to as a liquid period. In
the beginning of a liquid period only seasonal contracts are available.
Seasonal forward contracts are agreements between a buyer and a seller
that the seller will deliver a certain fixed amount of electricity
in every half hour during the season of interest at a price decided
upon today. As we get closer to the delivery, electricity contracts
with greater granularity appear. We can find quarterly, monthly, weekly,
daily, and intra day contracts. Intra day contracts can cover blocks
of twelve hours, four hours, two hours, one hour and the smallest
granularity is half an hour (see the APX power exchange%
\footnote{http://www.apxgroup.com/trading-clearing/apx-power-uk%
} for details). Furthermore, various combinations of the above, such
as month-ahead peak contracts, which cover all half hours between
7am and 7pm in the following month, are traded. One could use such
contracts to hedge the production of a solar power plant, for example.
As we can see there exist a huge variety of different contracts. Some
of them are traded through an exchange, while others are traded over-the-counter
(OTC). In the Wholesale Market Report%
\footnote{Report for July 2014 is available at http://www.energy-uk.org.uk/publication/finish/5-research-and-reports/1152-wholesale-market-report-july-2014.html.%
} from the Energy UK, we can see that a large majority of contracts
is traded OTC. When choosing the right contract to trade, one also
has to take into consideration liquidity and transaction costs. The
Wholesale Market Reports from Energy UK show that contracts with a
delivery after two or more years tend to be very illiquid. Baseload
contracts are much more liquid than peak (i.e. 7am to 7pm) or off-peak
(i.e. 7pm to 7am) contracts. The bid-ask spread for different types
of contracts between years 2008 and 2011 is available in the Ofgem
report%
\footnote{https://www.ofgem.gov.uk/ofgem-publications/39661/summer-2011-assessment.pdf%
}. Besides trading electricity through future and forward contracts,
a significant amount of electricity is also traded through day-ahead
auctions. According to the Wholesale Market Report from Energy UK,
roughly between 15\% and 20\% of power is traded through day-ahead
auctions offered by APX%
\footnote{http://www.apxgroup.com/trading-clearing/auction/%
} and N2EX%
\footnote{https://www.n2ex.com/%
}. Since the auctions are not explicitly included in our model, we
will not study them in further detail. An optimal bidding strategy
for the auction market is discussed in \cite{anderson2002necessary}
and \cite{anderson2005varepsilonoptimal}.

One hour before the delivery, at the event called a Gate Closure,
trading of future and forward contracts ceases. This is when the system
operator takes over the management of the power grid. All market participants
inform the system operator about their positions in future and forward
contracts. They also submit their bids and offers. A bid is a pair
of volume and a price which tells the system operator at what price
a producer can increase production. Similarly, the offer tell the
system operator what compensation the producer is willing to accept
from the operator if asked to decrease the production. Consumers with
a flexible consumption also submit bids and offers. Equipped with
this information the system operator first compares the demand forecast
with the submitted number of traded contracts and adjusts the difference
by accepting some of the bids or offers. While doing so, the system
operator must also take into account the capacity constraints of the
transmission lines in the network. After the delivery of electricity,
the system operator compares the actual physical production/consumption
of electricity by each market participant with the contracted volumes
adjusted by the accepted bids and offers and calculates the imbalance
volume as the difference between the two. If the imbalance volume
helped the system operator match supply and demand, then a fair, market
index price, is used to calculate the imbalance cash flow to/from
each player. On the other hand, if the imbalance volume hampered the
system operator in matching supply and demand, then a worse price%
\footnote{For a detailed calculation of this price see \\
http://www.elexon.co.uk/wp-content/uploads/2014/06/imbalance\_pricing\_guidance\_v7.0.pdf.%
} is applied to calculate the cash flow. Roughly speaking, this price
is calculated as the average of the most expensive 500 MWh of accepted
bids or offers adjusted by transmission losses.

The two modes of operation described above are very different. The
first mode of operation can be considered as a competitive market
without a central agent, while the second is controlled by the system
operator who is responsible for matching the electricity supply and
demand by choosing the cheapest actions. The first mode of operation
can be seen as relatively independent of the second mode and thus,
this work focuses on the first mode of the operation only. A coupling
of both modes is something that we would like to investigate in the
future.

\section{Problem description\label{sec:Problem-description}}

In this section we provide a detailed description of a model that
we use for the purpose of modeling the term structure of electricity
prices. The model belongs to a class of game theoretic equilibrium
models. Market participants are divided into consumers and producers.
A set of consumers is denoted by $C$ and has cardinality $0<\left|C\right|<\infty$.
Similarly, a set of producers is denoted by $P$ and has cardinality
$0<\left|P\right|<\infty$. Each producer owns a portfolio of power
plants that can have different characteristics such as capacity, ramp-up
and ramp-down constraints, efficiency, and fuel type. The set of all
fuel types is denoted by $L$. Sets $R^{p,l}$ denote all power plants
owned by producer $p\in P$ that run on fuel $l\in L$. A set $R^{p,l}$
may be empty since each producer typically does not own all possible
types of power plants. Moreover, this allows us to include non physical
traders such as banks or speculators, who do not own any electricity
generation facilities and are without a physical demand for electricity,
as producers $p\in P$ with $R^{p,l}=\left\{ \right\} $ for all $l\in L$. 

As we will see in Subsection \ref{sub:The-hypothetical-market}, it
is useful to introduce another player named the hypothetical market
agent besides producers and consumers. The hypothetical market agent
plays the role of the electricity market and ensures that the term
structure of the electricity price is such that the market clearing
condition is satisfied for all electricity forward contracts.

We are interested in delivery times $T_{j}$, $j\in J=\left\{ 1,...,T'\right\} $,
where power for each delivery time $T_{j}$ can be traded through
numerous forward contracts at times $t_{i}$, $i\in I_{j}$. The electricity
price at time $t_{i}$ for delivery at time $T_{j}$ is denoted by
$\Pi\left(t_{i},T_{j}\right)$. Since contracts with trading time
later than delivery time do not exist, we require $t_{\max\left\{ I_{j}\right\} }=T_{j}$
for all $j\in J$. The number of all forward contracts, i.e. $\sum_{j\in J}\left|I_{j}\right|$,
is denoted by $N$. Uncertainty is modeled by a filtered probability
space $\left(\Omega,\mathcal{F},\mathbb{F}=\left\{ \mathcal{F}_{t},t\in I\right\} ,\mathbb{P}\right)$,
where $I=\cup_{j\in J}I_{j}$. The $\sigma$-algebra $\mathcal{F}_{t}$
represents information available at time $t$.

The exogenous variables that appear in our model are (a) aggregate
power demand $D\left(T_{j}\right)$ for each delivery period $j\in J$,
(b) prices of fuel forward contracts $G_{l}\left(t_{i},T_{j}\right)$
for each fuel $l\in L$, delivery period $j\in J$, and trading period
$i\in I_{j}$, and (c) prices of emissions forward contracts $G_{em}\left(t_{i},T_{j}\right)$,
$j\in J$, $i\in I_{j}$. Electricity prices and all exogenous variables
are assumed to be adapted to the filtration $\left\{ \mathcal{F}_{t}\right\} _{t\in I}$
and have finite second moments.

Let $v_{k}\in\mathbb{R}^{n_{k}}$, $n_{k}\in\mathbb{N}$, $k\in K$,
and $K=\left\{ 1,...,\left|K\right|\right\} $ be given vectors. For
convenience, we define a vector concatenation operator as
\[
\left|\right|_{k\in K}v_{k}=\left[v_{1}^{\top},...,v_{\left|K\right|}^{\top}\right]^{\top}.
\]

\subsection{Producer}

Each producer $p\in P$ participates in the electricity, fuel, and
emission markets. Forward as well as spot contracts are available
on all markets. Electricity prices, fuel prices, and emission prices
are denoted by $\Pi\left(t_{i},T_{j}\right)$, $G_{l}\left(t_{i},T_{j}\right)$
where $l\in L$, and $G_{em}\left(t_{i},T_{j}\right)$, respectively.
To simplify the notation, we introduce
\begin{itemize}
\item electricity price vectors $\Pi\left(T_{j}\right)=\left|\right|_{i\in I_{j}}\Pi\left(t_{i},T_{j}\right)$,
and $\Pi=\left|\right|_{j\in J}e^{-\hat{r}T_{j}}\Pi\left(T_{j}\right)$,
where $\hat{r}\in\mathbb{R}$ is a constant interest rate, 
\item fuel price vectors $G\left(t_{i},T_{j}\right)=\left|\right|_{l\in L}G_{l}\left(t_{i},T_{j}\right)$,
$G\left(T_{j}\right)=\left|\right|_{i\in I_{j}}G\left(t_{i},T_{j}\right)$,
and\linebreak{}
$G=\left|\right|_{j\in J}e^{-\hat{r}T_{j}}G\left(T_{j}\right)$, 
\item emission price vector $G_{em}\left(T_{j}\right)=\left|\right|_{i\in I_{j}}G_{em}\left(t_{i},T_{j}\right)$,
and $G_{em}=\left|\right|_{j\in J}e^{-\hat{r}T_{j}}G_{em}\left(T_{j}\right)$. 
\end{itemize}
A producer may participate in the market by buying and selling forward
and spot contracts. The number of electricity forward contracts that
producer $p\in P$ buys at trading time $t_{i}$, $i\in I_{j}$ for
delivery at time $T_{j}$, $j\in J$ is denoted by $V_{p}\left(t_{i},T_{j}\right)$.
Similarly, the number of fuel and emission forward contracts that
producer $p\in P$ buys at trading time $t_{i}$, $i\in I_{j}$ for
delivery at time $T_{j}$, $j\in J$ is denoted by $F_{p,l}\left(t_{i},T_{j}\right)$,
$l\in L$ and $O_{p}\left(t_{i},T_{j}\right)$, respectively. Producers
own a generally non-empty portfolio of power plants. The actual production
of electricity from power plant $r\in R^{p,l}$ at delivery time $T_{j}$,
$j\in J$ is denoted by $W_{p,l,r}\left(T_{j}\right)$.

The notation is greatly simplified if the decision variables are concatenated
into 
\begin{itemize}
\item electricity trading vectors $V_{p}\left(T_{j}\right)=\left|\right|_{i\in I_{j}}V_{p}\left(t_{i},T_{j}\right)$
and $V_{p}=\left|\right|_{j\in J}V_{p}\left(T_{j}\right)$, 
\item fuel trading vectors $F_{p}\left(t_{i},T_{j}\right)=\left|\right|_{l\in L}F_{p,l}\left(t_{i},T_{j}\right)$,
$F_{p}\left(T_{j}\right)=\left|\right|_{i\in I_{j}}F_{p}\left(t_{i},T_{j}\right)$,
and\linebreak{}
 $F_{p}=\left|\right|_{j\in J}F_{p}\left(T_{j}\right)$, 
\item emission trading vectors $O_{p}\left(T_{j}\right)=\left|\right|_{i\in I_{j}}O_{p}\left(t_{i},T_{j}\right)$
and $O_{p}=\left|\right|_{j\in J}O_{p}\left(T_{j}\right)$, 
\item electricity production vectors $W_{p,l}\left(T_{j}\right)=\left|\right|_{r\in R^{p,l}}W_{p,l,r}\left(T_{j}\right)$,
$W_{p}\left(T_{j}\right)=\left|\right|_{l\in L}W_{p,l}\left(T_{j}\right)$,
and $W_{p}=\left|\right|_{j\in J}W_{p}\left(T_{j}\right)$, 
\end{itemize}
and finally $v_{p}=\left[V_{p}^{\top},F_{p}^{\top},O_{p}^{\top},W_{p}^{\top}\right]^{\top}$.
Similarly, we concatenate all price vectors as 
\[
\pi_{p}=\left[\Pi^{\top},G^{\top},G_{em}^{\top},\underset{\dim\left(W_{p}\right)}{\underbrace{0,...,0}}\right]^{\top},
\]
where the number of zeros matches the dimension of the vector $W_{p}$. 

Producer $p\in P$ is not able to arbitrarily choose her decision
variables since there are some inequality and equality constraints
that limit her feasible set. The change in production of each power
plant from one delivery period to next is limited by the ramp up and
ramp down constraints. For each $j\in\left\{ 1,...,T'-1\right\} $,
where $T'$ denotes the last delivery period, $l\in L$ and $r\in R^{p,l}$
these constraints can be expressed as 
\begin{equation}
\triangle\overline{W}_{min}^{p,l,r}\leq W_{p,l,r}\left(T_{j+1}\right)-W_{p,l,r}\left(T_{j}\right)\leq\triangle\overline{W}_{max}^{p,l,r},\label{eq:pb0}
\end{equation}
where $\triangle\overline{W}_{max}^{p,l,r}$ and $\triangle\overline{W}_{min}^{p,l,r}$
represent maximum rates for ramping up and down, respectively. The
ramping rates highly depend on the type of the power plant. Some gas
power plants can increase production from zero to the maximum in just
a few minutes, while the same action may take days or weeks for a
nuclear power plant. 

The capacity constraints for each power plant $r\in R^{p,l}$ can
be expressed as
\begin{equation}
0\leq W_{p,l,r}\left(T_{j}\right)\leq\overline{W}_{max}^{p,l,r},\label{eq:pb1}
\end{equation}
where $\overline{W}_{max}^{p,l,r}$ denotes the maximum production. 

Additionally, we also bound the the number of electricity contracts
that each player is allowed to trade as 
\begin{equation}
-V_{trade}\leq V_{p}\left(t_{i},T_{j}\right)\leq V_{trade}\label{eq:pb2}
\end{equation}
for some large $V_{trade}>0$. Trading of an infinite number of contracts
would clearly lead to a bankruptcy of one of the counterparties involved
and must thus be prevented. In \cite{troha2014theexistence} it was
shown, that if $V_{trade}$ is chosen to be large enough, then Constraint
(\ref{eq:pb2}) has no impact on the optimal solution and can be eliminated
from the problem. 

There are also equality constraints that connect power plant production
with electricity, fuel, and emission trading. For each $j\in J$ the
electricity sold in the forward and spot market together must equal
the actually produced electricity, i.e.
\begin{equation}
-\sum_{i\in I_{j}}V_{p}\left(t_{i},T_{j}\right)=\sum_{l\in L}\sum_{r\in R^{p,l}}W_{p,l,r}\left(T_{j}\right).\label{eq:pb2-1}
\end{equation}
Each producer $p\in P$ has to make sure that a sufficient amount
of fuel $l\in L$ has been bought to cover the electricity production
for each delivery period $j\in J$ . Such constraint can be expressed
as 
\begin{equation}
\sum_{r\in R^{p,l}}W_{p,l,r}\left(T_{j}\right)c^{p,l,r}=\sum_{i\in I_{j}}F_{p,l}\left(t_{i},T_{j}\right)\label{eq:pb3-1}
\end{equation}
where $c^{p,l,r}>0$ is the efficiency of power plant $r\in R^{p,l}$. 

The carbon emission obligation constraint can be written as
\begin{equation}
\sum_{j\in J}\sum_{i\in I_{j}}O\left(t_{i},T_{j}\right)=\sum_{j\in J}\sum_{l\in L}\sum_{r\in R^{p,l}}W_{p,l,r}\left(T_{j}\right)g^{p,l,r},\label{eq:pb4}
\end{equation}
where $g^{p,l,r}>0$ denotes the carbon emission intensity factor
for power plant $r\in R^{p,l}$. This constraint ensures that enough
emission certificates have been bought to cover the electricity production
over the whole planning horizon. 

Any producers' goal is to maximize their expected profit subject to
a risk budget. In this work we assume that the risk budget is expressed
in a mean-variance framework. As mentioned in the introduction, the
main argument that supports this decision is that delta hedging, which
is the most widely used hedging strategy, can be captured in this
framework. 

The profit $P_{p}\left(v_{p},\pi_{p}\right)$ of producer $p\in P$
can be calculated as
\begin{equation}
P_{p}\left(v_{p},\pi_{p}\right)=\sum_{j\in J}e^{-\hat{r}T_{j}}\left(\sum_{i\in I_{j}}P_{p}^{t_{i},T_{j}}\left(v_{p},\pi_{p}\right)\right)
\end{equation}
 where the profit $P_{p}^{t_{i},T_{j}}\left(v_{p},\pi_{p}\right)$
for each $i\in I_{j}$ and $j\in J$ can be calculated as

\[
P_{p}^{t_{i},T_{j}}\left(v_{p},\pi_{p}\right)=-\Pi\left(t_{i},T_{j}\right)V_{p}\left(t_{i},T_{j}\right)-O_{p}\left(t_{i},T_{j}\right)G_{em}\left(t_{i},T_{j}\right)-\sum_{l\in L}G_{l}\left(t_{i},T_{j}\right)F_{p,l}\left(t_{i},T_{j}\right).
\]
Under a mean-variance optimization framework, producers are interested
in the mean-variance utility

\[
\begin{array}{rcl}
\Psi_{p}\left(v_{p}\right) & = & \mathbb{E}^{\mathbb{P}}\left[P_{p}\left(v_{p},\pi_{p}\right)\right]-\frac{\lambda_{p}}{2}\text{Var}^{\mathbb{P}}\left[P_{p}\left(v_{p},\pi_{p}\right)\right]\\
\\
 & = & -\mathbb{E}^{\mathbb{P}}\left[\pi_{p}\right]^{\top}v_{p}-\frac{1}{2}\lambda_{p}v_{p}^{\top}Q_{p}v_{p},
\end{array}
\]
where $\lambda_{p}>0$ is their risk preference parameter and $Q_{p}:=\mathbb{E}^{\mathbb{P}}\left[\left(\pi_{p}-\mathbb{E}^{\mathbb{P}}\left[\pi_{p}\right]\right)\left(\pi_{p}-\mathbb{E}^{\mathbb{P}}\left[\pi_{p}\right]\right)^{\top}\right]$
an ``extended'' covariance matrix. Their objective is to solve the
following optimization problem
\begin{equation}
\begin{array}{rl}
\Phi_{p}=\underset{v_{p}}{\text{max }} & \Psi_{p}\left(v_{p}\right)\end{array}\label{eq:opt_prod}
\end{equation}
subject to (\ref{eq:pb0}), (\ref{eq:pb1}), (\ref{eq:pb2}), (\ref{eq:pb2-1}),
(\ref{eq:pb3-1}), and (\ref{eq:pb4}).

\subsection{Consumer}

We make the assumption that demand is completely inelastic and that
each consumer $c\in C$ is responsible for satisfying a proportion
$p_{c}\in\left[0,1\right]$ of the total demand $D\left(T_{j}\right)$
at time $T_{j}$, $j\in J$. Since $p_{c}$ is a proportion, we clearly
have that $\sum_{c\in C}p_{c}=1.$ 

A number of electricity forward contracts consumer $c\in C$ buys
at trading time $t_{i}$, $i\in I_{j}$ for delivery at time $T_{j}$,
$j\in J$ is denoted by $V_{c}\left(t_{i},T_{j}\right)$. Similarly
as for producers, we can simplify the notation by introducing electricity
trading vectors $V_{c}\left(T_{j}\right)=\left|\right|_{i\in I_{j}}V_{c}\left(t_{i},T_{j}\right)$
and $V_{c}=\left|\right|_{j\in J}V_{c}\left(T_{j}\right)$. 

We bound the the number of electricity contracts that each player
is allowed to trade as 
\begin{equation}
-V_{trade}\leq V_{p}\left(t_{i},T_{j}\right)\leq V_{trade}\label{eq:pb2-2}
\end{equation}
for some large $V_{trade}>0$. Trading of an infinite number of contracts
would clearly lead to a bankruptcy of one of the counterparties involved
and must thus be prevented. In \cite{troha2014theexistence} it was
shown, that if $V_{trade}$ is chosen large enough, then Constraint
(\ref{eq:pb2-2}) has no impact on the optimal solution and can be
eliminated from the problem. 

Consumers are responsible for satisfying the electricity demand of
end users. The electricity demand is expected to be satisfied for
each $T_{j}$, i.e.
\begin{equation}
\sum_{i\in I_{j}}V_{c}\left(t_{i},T_{j}\right)=p_{c}D\left(T_{j}\right).\label{eq:cb2}
\end{equation}
At the time of calculating the optimal decisions, consumers assume
that they know the future realization of demand $D\left(T_{j}\right)$
precisely. If the knowledge about the future realization of the demand
changes, then players can take recourse actions by recalculating their
optimal decisions with the updated demand forecast. Consumers may
assume that they will be able to execute the recourse actions, because
it is the job of the system operator to ensure that a sufficient amount
of electricity is available on the market.

Consumers would like to maximize their profit subject to a risk budget.
Similar to the model we introduced for producers, we assume that the
risk budget can be expressed in a mean-variance framework. The profit
of consumer $c\in C$ can be calculated as
\begin{equation}
P_{c}\left(V_{c},\Pi\right)=\sum_{j\in J}e^{-\hat{r}T_{j}}\left(\sum_{i\in I_{j}}-\Pi\left(t_{i},T_{j}\right)V_{c}\left(t_{i},T_{j}\right)+s_{c}p_{c}D\left(T_{j}\right)\right),\label{eq:23}
\end{equation}
where $\hat{r}\in\mathbb{R}$ denotes a constant interest rate and
$s_{c}\in\mathbb{R}$ denotes a contractually fixed price that consumer
$c\in C$ receives for selling the electricity further to end users
(e.g. households, businesses etc.). Note that the contractually fixed
price $s_{c}$ only affects the optimal objective value of consumer
$c\in C$, but not also her optimal solution. Since we are primarily
interested in optimal solutions, we simplify the notation and set
$s_{c}=0$. The correct optimal value can always be calculated via
post-processing when an optimal solution is already known. This may
be needed for risk management purposes.

Under a mean-variance optimization framework consumers are interested
in the mean-variance utility

\[
\begin{array}{rcl}
\Psi_{c}\left(V_{c}\right) & = & \mathbb{E}^{\mathbb{P}}\left[P_{c}\left(V_{c},\Pi\right)\right]-\frac{\lambda_{p}}{2}\text{Var}^{\mathbb{P}}\left[P_{c}\left(V_{c},\Pi\right)\right]\\
\\
 & = & -\mathbb{E}^{\mathbb{P}}\left[\Pi\right]^{\top}V_{c}-\frac{\lambda}{2}V_{c}^{\top}Q_{c}V_{c},
\end{array}
\]
where $\lambda_{c}>0$ is their risk preference and $Q_{c}:=\mathbb{E}^{\mathbb{P}}\left[\left(\Pi-\mathbb{E}^{\mathbb{P}}\left[\Pi\right]\right)\left(\Pi-\mathbb{E}^{\mathbb{P}}\left[\Pi\right]\right)^{\top}\right]$
a covariance matrix. Their objective is to solve the following optimization
problem
\begin{equation}
\Phi_{c}=\underset{V_{c}}{\text{max }}\Psi_{c}\left(V_{c}\right)\label{eq:opt_con}
\end{equation}
subject to (\ref{eq:pb2-2}) and (\ref{eq:cb2}).

\subsection{Matrix notation}

The analysis of the problem is greatly simplified if a more compact
notation is introduced.

Equality constraints of producer $p\in P$ can be expressed as
\[
A_{p}v_{p}=0
\]
and inequality constraints as 
\[
B_{p}v_{p}\leq b_{p}
\]
for some $A_{p}\in\mathbb{R}^{\left|J\right|\left(\left|L\right|+1\right)+1\times\dim v_{p}}$,
$B_{p}\in\mathbb{R}^{n_{p}\times\dim v_{p}}$ and $b_{p}\in\mathbb{R}^{n_{p}}$,
where $n_{p}$ denotes the number of the inequality constraints of
producer $p\in P$. Define a feasible set 
\[
S_{p}:=\left\{ v_{p}:A_{p}v_{p}=a_{p}\text{ \text{and }}B_{p}v_{p}\leq b_{p}\right\} .
\]
It is useful to investigate the inner structure of the matrices. By
considering equality constraints (\ref{eq:pb2-1}), (\ref{eq:pb3-1}),
and (\ref{eq:pb4}) we can see that 
\begin{equation}
A_{p}=\left[\begin{array}{ccc}
\hat{A}_{1} & 0 & \hat{A}_{3,p}\\
0 & \hat{A}_{2} & \hat{A}_{4,p}
\end{array}\right]\label{eq:Ap}
\end{equation}
where $\hat{A}_{1}\in\mathbb{R}^{\left|J\right|\times N},\hat{A}_{2}\in\mathbb{R}^{\left(\left|J\right|\left|L\right|+1\right)\times N\left(\left|L\right|+1\right)},\hat{A}_{3,p}\in\mathbb{R}^{\left|J\right|\times\dim W_{p}},\hat{A}_{4,p}\in\mathbb{R}^{\left(\left|J\right|\left|L\right|+1\right)\times\dim W_{p}}$.
One can see that matrices $\hat{A}_{1}$ and $\hat{A}_{2}$ are independent
of producer $p\in P$ and matrices $\hat{A}_{3,p}$ and $\hat{A}_{4,p}$
depend on producer $p\in P$. One can further investigate the structure
of $\hat{A}_{1}$ and see 
\begin{equation}
\hat{A}_{1}=\left[\begin{array}{ccc}
1_{1} &  & 0\\
 & \ddots\\
0 &  & 1_{\left|J\right|}
\end{array}\right],\label{eq:A1}
\end{equation}
where $1_{j}$, $j\in J$ is a row vector of ones of length $\left|I_{j}\right|$.
Similarly,
\begin{equation}
\hat{A}_{2}=\left[\begin{array}{cccc}
\hat{A}_{1} & \cdots & 0 & 0\\
\vdots & \ddots & \vdots & \vdots\\
0 & \cdots & \hat{A}_{1} & 0\\
0 & \cdots & 0 & 1_{\left|N\right|}
\end{array}\right],\label{eq:A2}
\end{equation}
where the number of rows in the block notation above is $\left|L\right|+1$.
The first $\left|L\right|$ rows correspond to (\ref{eq:pb3-1}) and
the last row corresponds to (\ref{eq:pb4}). 

The profit of producer $p\in P$ can be written as
\[
P_{p}\left(v_{p},\pi_{p}\right)=-\pi_{p}^{\top}v_{p}.
\]
In a compact notation, the mean-variance utility of producer $p\in P$
can be calculated as
\[
\begin{array}{rcl}
\Psi_{p}\left(v_{p},\mathbb{E}^{\mathbb{P}}\left[\Pi\right]\right) & = & \mathbb{E}^{\mathbb{P}}\left[-\pi_{p}^{\top}v_{p}-\frac{1}{2}\lambda_{p}v_{p}^{\top}\left(\pi_{p}-\mathbb{E}^{\mathbb{P}}\left[\pi_{p}\right]\right)\left(\pi_{p}-\mathbb{E}^{\mathbb{P}}\left[\pi_{p}\right]\right)^{\top}v_{p}\right]\\
\\
 & = & -\mathbb{E}^{\mathbb{P}}\left[\pi_{p}\right]^{\top}v_{p}-\frac{1}{2}\lambda_{p}v_{p}^{\top}Q_{p}v_{p},
\end{array}
\]
where
\begin{equation}
Q_{p}:=\mathbb{E}^{\mathbb{P}}\left[\left(\pi_{p}-\mathbb{E}^{\mathbb{P}}\left[\pi_{p}\right]\right)\left(\pi_{p}-\mathbb{E}^{\mathbb{P}}\left[\pi_{p}\right]\right)^{\top}\right].
\end{equation}
The inner structure of matrix $Q_{p}$ is the following 
\begin{equation}
Q_{p}=\left[\begin{array}{ccc}
\hat{Q}_{1} & \hat{Q}_{2} & 0\\
\hat{Q}_{2}^{\top} & \hat{Q}_{3} & 0\\
0 & 0 & 0
\end{array}\right]
\end{equation}
where $\hat{Q}_{1}\in\mathbb{R}^{N\times N},\hat{Q}_{2}\in\mathbb{R}^{N\times\left(\dim B_{p}+\dim O_{p}\right)}=\mathbb{R}^{N\times N\left(\left|L\right|+1\right)},\hat{Q}_{3}\in\mathbb{R}^{N\left(\left|L\right|+1\right)\times N\left(\left|L\right|+1\right)}$.
One can see that $\hat{Q}_{1}$, $\hat{Q}_{2}$, and $\hat{Q}_{3}$
do not depend on producer $p\in P$. The size of the larger matrix
$Q_{p}$ depends on producer $p\in P$, because different producers
have different number of power plants. 

Producer $p\in P$ attempts to solve the following optimization problem
\[
\Phi_{p}\left(\mathbb{E}^{\mathbb{P}}\left[\Pi\right]\right)=\underset{v_{p}\in S_{p}}{\text{max }}-\mathbb{E}^{\mathbb{P}}\left[\pi_{p}\right]^{\top}v_{p}-\frac{1}{2}\lambda_{p}v_{p}^{\top}Q_{p}v_{p}.
\]

The equality constraints of consumer $c\in C$ can be expressed as
\[
A_{c}V_{c}=a_{c}
\]
and the inequality constraints as 
\[
B_{c}V_{c}\leq b_{c}
\]
where $A_{c}=\hat{A}_{1}$, $B_{c}\in\mathbb{R}^{2N\times N}$, $a_{c}\in\mathbb{R}^{\left|J\right|}$
and $b_{c}\in\mathbb{R}^{2N}$. Define a feasible set
\[
S_{c}:=\left\{ V_{c}\in\mathbb{R}^{N}:A_{c}V_{c}=a_{c}\text{ \text{and }}B_{c}V_{c}\leq b_{c}\right\} .
\]

The profit of consumer $c\in C$ can be written as
\[
P_{c}\left(V_{c},\Pi\right)=-\Pi^{\top}V_{c}.
\]
In a compact notation, the mean-variance utility of a consumer $c\in C$
can be calculated as
\[
\begin{array}{rcl}
\Psi_{c}\left(V_{c},\mathbb{E}^{\mathbb{P}}\left[\Pi\right]\right) & = & \mathbb{E}^{\mathbb{P}}\left[-\Pi^{\top}V_{c}-\frac{1}{2}\lambda_{c}V_{c}^{\top}\left(\Pi-\mathbb{E}^{\mathbb{P}}\left[\Pi\right]\right)\left(\Pi-\mathbb{E}^{\mathbb{P}}\left[\Pi\right]\right)^{\top}V_{c}\right]\\
\\
 & = & -\mathbb{E}^{\mathbb{P}}\left[\Pi\right]^{\top}V_{c}-\frac{\lambda}{2}V_{c}^{\top}Q_{c}V_{c},
\end{array}
\]
where
\begin{equation}
Q_{c}:=\mathbb{E}^{\mathbb{P}}\left[\left(\Pi-\mathbb{E}^{\mathbb{P}}\left[\Pi\right]\right)\left(\Pi-\mathbb{E}^{\mathbb{P}}\left[\Pi\right]\right)^{\top}\right].
\end{equation}
 Moreover, note that $Q_{c}=\hat{Q}_{1}$ for all $c\in C$. We set
$s_{c}=0$, w.l.o.g. Consumer $c\in C$ attempts to solve the following
optimization problem
\[
\Phi_{c}\left(\mathbb{E}^{\mathbb{P}}\left[\Pi\right]\right)=\underset{V_{c}\in S_{c}}{\text{max }}-\mathbb{E}^{\mathbb{P}}\left[\Pi\right]^{\top}V_{c}-\frac{\lambda}{2}V_{c}^{\top}Q_{c}V_{c}.
\]

\subsection{The hypothetical market agent\label{sub:The-hypothetical-market}}

Given the price vectors of electricity $\Pi$, fuel $G$, and emissions
$G_{em}$, each producer $p\in P$ and each consumer $c\in C$ can
calculate their optimal electricity trading vectors $V_{p}$ and $V_{c}$
by solving (\ref{eq:opt_prod}) and (\ref{eq:opt_con}), respectively.
However, the players are not necessary able to execute their calculated
optimal trading strategies because they may not find the counterparty
to trade with. In reality each contract consists of a buyer and a
seller, which imposes an additional constraint (also called the market
clearing constraint) that matches the number of short and long electricity
contracts for each $i\in I_{j}$ and $j\in J$ as follows,
\begin{equation}
\sum_{c\in C}V_{c}\left(t_{i},T_{j}\right)+\sum_{p\in P}V_{p}\left(t_{i},T_{j}\right)=0.\label{eq:mk}
\end{equation}
The electricity market is responsible for satisfying this constraint
by matching buyers with sellers. The matching is done through sharing
of the price and order book information among all market participants.
If at the current price there are more long contract than short contracts,
it means that the current price is too low and asks will start to
be submitted at higher prices. The converse occurs, if there are more
short contracts than long contracts. Eventually, the electricity price
at which the number of long and short contracts matches is found.
At such a price the constraint (\ref{eq:mk}) is satisfied ``naturally''
without explicitly requiring the players to satisfy it. They do so
because it is in their best interest, i.e. it maximizes their mean-variance
objective functions.

The question is how to formulate such an equilibrium constraint in
an optimization framework. A naive approach of writing the market
clearing constraint as an ordinary constraint forces the players to
satisfy it regardless of the price. We need a mechanism that models
the matching of buyers and sellers as it is performed by the electricity
market. For this purpose, we introduce a hypothetical market agent
who is allowed to slowly change electricity prices to ensure that
(\ref{eq:mk}) is satisfied.

Let the hypothetical market agent have the following profit function
\begin{equation}
\begin{array}{rcl}
P_{M}\left(\Pi,V\right) & = & \sum_{j\in J}e^{-\hat{r}T_{j}}\left[\sum_{i\in I_{j}}\Pi\left(t_{i},T_{j}\right)\left(\sum_{c\in C}V_{c}\left(t_{i},T_{j}\right)+\sum_{p\in P}V_{p}\left(t_{i},T_{j}\right)\right)\right]\\
\\
 & = & \mathbb{E}^{\mathbb{P}}\left[\Pi\right]^{\top}\left(\sum_{c\in C}V_{c}+\sum_{p\in P}V_{p}\right)
\end{array}\label{eq:market-1}
\end{equation}
and the expected profit
\begin{equation}
\Psi_{M}\left(\mathbb{E}^{\mathbb{P}}\left[\Pi\right],V\right)=\mathbb{E}^{\mathbb{P}}\left[P_{M}\left(V,\Pi\right)\right],\label{eq:utility_market}
\end{equation}
where $V=\left[V_{P}^{\top},V_{C}^{\top}\right]^{\top}$, $V_{P}=\left|\right|_{p\in P}V_{p}$,
and $V_{C}=\left|\right|_{c\in C}V_{c}$ and let the hypothetical
market agent attempts to solve
\begin{equation}
\Phi_{M}\left(V\right)=\underset{\mathbb{E}^{\mathbb{P}}\left[\Pi\right]}{\text{max }}\Psi_{M}\left(\mathbb{E}^{\mathbb{P}}\left[\Pi\right],V\right).\label{eq:SO}
\end{equation}
The KKT conditions for (\ref{eq:SO}) in the matrix notation read
\begin{equation}
\sum_{c\in C}V_{c}+\sum_{p\in P}V_{p}=0,\label{eq:KKT_hyp_mark}
\end{equation}
which is exactly the same as (\ref{eq:mk}). Note, that the equivalence
of (\ref{eq:mk}) and (\ref{eq:SO}) is a theoretical result that
has to be applied with caution in an algorithmic framework. Formulation
(\ref{eq:SO}) is clearly unstable since only a small mismatch in
the market clearing constraint sends the prices to $\pm\infty$. Thus,
a stable formulation of the hypothetical market agent must be found.
Let us now analyze the hypothetical market agent with the following,
slightly altered, optimization problem
\begin{equation}
\begin{array}{rl}
\underset{\mathbb{E}^{\mathbb{P}}\left[\Pi\right]}{\text{max }} & \Psi_{M}\left(\mathbb{E}^{\mathbb{P}}\left[\Pi\right],V\right)\\
\\
\text{s.t.} & \sum_{c\in C}V_{c}+\sum_{p\in P}V_{p}=0\\
\\
 & \mu_{M}=0,
\end{array}\label{eq:market_new}
\end{equation}
where $\mu_{M}$ denotes the dual variables of the equality constraint
in (\ref{eq:market_new}). It is trivial to check that the optimality
conditions for (\ref{eq:market_new}) correspond to (\ref{eq:mk}).
Formulation (\ref{eq:market_new}) is clearly stable, because the
market clearing constraint is satisfied precisely. The equality constraint
on the dual variables makes sure that the optimal solution remains
the same if the market clearing constraint is removed after the calculation
of the optimal solution. Formulation (\ref{eq:market_new}) is used
as a definition of the hypothetical market agent in the rest of this
work.

We can see that, by affecting the expected electricity price, the
hypothetical agent changes the electricity price process. It is not
immediately clear how to construct such a stochastic process or that
such a stochastic process exists at all. We refer the reader to \cite{troha2014theexistence},
where a constructive proof of the existence is given. The proof is
based on the Doob decomposition theorem, where we allow the hypothetical
market agent to control an integrable predictable term of the process,
while keeping the martingale term of the process intact.

For the further argumentation we define $v_{P}=\left|\right|_{p\in P}v_{p}$
and $v=\left[v_{P}^{\top},V_{C}^{\top}\right]^{\top}$.

\subsection{Nash equilibrium\label{sec:Analysis}}

We are interested in finding a Nash equilibrium defined as

\begin{definition}\label{Nash-Equilibrium-(NE)}Nash Equilibrium
(NE)

Decisions $v^{*}$ and $\mathbb{E}^{\mathbb{P}}\left[\Pi\right]^{*}$
constitute a Nash equilibrium if
\begin{enumerate}
\item For every producer $p\in P$, $v_{p}^{*}$ is a strategy such that
\begin{equation}
\Psi_{p}\left(v_{p},\mathbb{E}^{\mathbb{P}}\left[\Pi\right]^{*}\right)\leq\Psi_{p}\left(v_{p}^{*},\mathbb{E}^{\mathbb{P}}\left[\Pi\right]^{*}\right)
\end{equation}
for all $v_{p}\in S_{p}$;
\item For every consumer $c\in C$, $V_{c}^{*}$ is a strategy such that
\begin{equation}
\Psi_{c}\left(V_{c},\mathbb{E}^{\mathbb{P}}\left[\Pi\right]^{*}\right)\leq\Psi_{c}\left(V_{c}^{*},\mathbb{E}^{\mathbb{P}}\left[\Pi\right]^{*}\right)
\end{equation}
for all $V_{c}\in S_{c}$;
\item Price vector $\mathbb{E}^{\mathbb{P}}\left[\Pi\right]^{*}$ maximizes
the objective function of the hypothetical market agent, i.e. 
\begin{equation}
\Psi_{M}\left(\mathbb{E}^{\mathbb{P}}\left[\Pi\right],v^{*}\right)\leq\Psi_{M}\left(\mathbb{E}^{\mathbb{P}}\left[\Pi\right]^{*},v^{*}\right)\label{eq:market}
\end{equation}
for all $\mathbb{E}^{\mathbb{P}}\left[\Pi\right]\in S_{M}$.
\end{enumerate}
\end{definition}

From Definition (\ref{Nash-Equilibrium-(NE)}), it is not clear whether
a NE for our problem exists and whether it is unique. This problem
was thoroughly investigated in \cite{troha2014theexistence}. Roughly
speaking, it was shown that if the demand of the end users can be
covered by the available system of power plants, then a NE exists.
Moreover, if the power plants are similar enough (if there are no
big gaps in the efficiency of the power plants), then one can show
that the NE is also unique. On the other hand, if power plants are
similar enough, then the expected equilibrium price of each electricity
contract might be an interval instead of a single point. 

In this paper we focus on the numerical calculation of the NE under
the assumption of the existence of solution. For this paper, we assume
the following a slightly stricter condition.

\begin{assumption}\label{ass: as1} For all $p\in P$, the exists
vector $v_{p}$ such that $A_{p}v_{p}=a_{p}$ a.s. and $B_{p}v_{p}<b_{p}$
a.s., for all $c\in C$, there exists vector $V_{c}$ such that $A_{c}V_{c}=a_{c}$
a.s. and $B_{c}V_{c}<b_{c}$ a.s., and the vectors $V_{p}$ and $V_{c}$
can be chosen so that (\ref{eq:KKT_hyp_mark}) is satisfied.\end{assumption}

\section{Quadratic programming formulation\label{sec:Algorithm}}

The traditional approach to solving equilibrium optimization problems
is through shadow prices (see \cite{demaeredaertrycke2012liquidity}
for example). However, this approach is only valid when no inequality
constraints are present. Shadow prices depend on the set of active
constraints and thus one can only use this approach when the active
set is known. In inequality constrainted optimization, the active
set is usually not know in advance and thus a different approach is
needed. The proposed formulation below can be seen as an extension
of the shadow price concept to inequality constrained optimization
problems. 

A naive approach for solving inequality constrained equilibrium optimization
problem would be to choose an expected price vector $\mathbb{E}^{\mathbb{P}}\left[\Pi\right]$
and then calculate optimal solutions for each producer $p\in P$ and
each consumer $c\in C$. If at such price $\left\Vert \sum_{c\in C}V_{c}+\sum_{p\in P}V_{p}\right\Vert $
is close to zero, then the solution is found and $\mathbb{E}^{\mathbb{P}}\left[\Pi\right]$
is an equilibrium expected price vector. Otherwise, we have to adjust
the expected price vector and repeat the procedure. We can see that
such an algorithm is costly, because it requires to solve a large
optimization problem (i.e. to calculate the optimal solutions of each
producer and each consumer) multiple times. In the section below,
we show that we can do much better than the naive approach. Using
the reformulation we propose, the large optimization problem must
be solved only once. 

Necessary and sufficient conditions for all $v_{k}$, $k\in P\cup C$
and $\mathbb{E}^{\mathbb{P}}\left[\Pi\right]$ to constitute a NE
are the following, due to the fact that Assumption \ref{ass: as1}
implies the Slater condition, 
\begin{equation}
\begin{array}{rcl}
-\mathbb{E}^{\mathbb{P}}\left[\pi_{k}\right]^{\top}-\lambda_{k}Q_{k}v_{k}-B_{k}^{\top}\eta_{k}-A_{k}^{\top}\mu_{k} & = & 0\\
\\
\eta_{k}^{\top}\left(B_{k}v_{k}-b_{k}\right) & = & 0\\
\\
B_{k}v_{k}-b_{k} & \leq & 0\\
\\
A_{k}v_{k}-a_{k} & = & 0\\
\\
\eta_{k} & \geq & 0\\
\\
\sum_{c\in C}V_{c}+\sum_{p\in P}V_{p} & = & 0.
\end{array}\label{eq:KKT_all-2}
\end{equation}
The last equation corresponds to the KKT conditions of the hypothetical
market agent. 

We can now interpret (\ref{eq:KKT_all-2}) as the KKT conditions of
one large optimization problem that includes the new definition (\ref{eq:market_new})
of the hypothetical market agent. To see this, we join all decision
variables into one vector $x:=\left[v^{\top},\mathbb{E}^{\mathbb{P}}\left[\Pi\right]^{\top}\right]^{\top}$
and rewrite
\begin{itemize}
\item the equality constraints as $Ax=a$ with $a:=\left[a_{p_{1}}^{\top},...,a_{p_{P}}^{\top},a_{c_{1}}^{\top},...,a_{c_{C}}^{\top},\underbrace{0,...,0}_{N}\right]^{\top}$where
the number of ending zeros is equal to $N$, and 
\[
A:=\left[\begin{array}{ccccccc}
A_{p_{1}} & 0 &  &  &  &  & 0\\
0 & \ddots & 0 &  &  &  & \vdots\\
 & 0 & A_{p_{P}} & 0 &  &  & 0\\
 &  & 0 & A_{c_{1}} & 0 &  & 0\\
 &  &  & 0 & \ddots & 0 & \vdots\\
 &  &  &  & 0 & A_{c_{C}} & 0\\
M_{p_{1}} & \cdots & M_{p_{P}} & I & \cdots & I & 0
\end{array}\right],
\]
where $M_{p}\in\mathbb{R}^{N}\times\mathbb{R}^{\dim v_{p}}$ is a
matrix defined as 
\[
M_{p}=\left[\text{diag}\left(\underset{N}{\underbrace{1,...,1}}\right)\left|\begin{array}{ccc}
0 & \cdots & 0\\
\vdots & \ddots & \vdots\\
0 & \cdots & 0
\end{array}\right.\right],
\]

\item the inequality constraints as $Bx\leq b$ with $b:=\left[b_{p_{1}}^{\top},...,b_{p_{P}}^{\top},b_{c_{1}}^{\top},...,b_{c_{C}}^{\top}\right]^{\top}$,
and
\[
B:=\left[\begin{array}{ccccccc}
B_{p_{1}} & 0 &  &  &  &  & 0\\
0 & \ddots & 0 &  &  &  & \vdots\\
 & 0 & B_{p_{P}} & 0 &  &  & 0\\
 &  & 0 & B_{c_{1}} & 0 &  & 0\\
 &  &  & 0 & \ddots & 0 & \vdots\\
 &  &  &  & 0 & B_{c_{C}} & 0
\end{array}\right],
\]

\item the objective function as $-\pi^{\top}x-\frac{1}{2}x^{\top}Qx$ with
$\pi:=\left[\mathbb{E}^{\mathbb{P}}\left[\pi_{0,p_{1}}\right]^{\top},...,\mathbb{E}^{\mathbb{P}}\left[\pi_{0,p_{P}}\right]^{\top},\underset{\left(\left|C\right|+1\right)N}{\underbrace{0,...,0}}\right]^{\top}$
where $\pi_{0,p}$ is $\pi_{p}$ with elements of $\Pi$ set to zero,
and
\begin{equation}
Q:=\left[\begin{array}{ccccccc}
\lambda_{p_{1}}Q_{p_{1}} & 0 &  &  &  &  & M_{p_{1}}^{\top}\\
0 & \ddots & 0 &  &  &  & \vdots\\
 & 0 & \lambda_{p_{P}}Q_{p_{P}} & 0 &  &  & M_{p_{P}}^{\top}\\
 &  & 0 & \lambda_{c_{1}}Q_{c_{1}} & 0 &  & I\\
 &  &  & 0 & \ddots & 0 & \vdots\\
 &  &  &  & 0 & \lambda_{c_{C}}Q_{c_{C}} & I\\
M_{p_{1}} & \cdots & M_{p_{P}} & I & \cdots & I & 0
\end{array}\right],\label{eq:Q-1}
\end{equation}

\item the dual variables as $\eta:=\left[\eta_{p_{1}}^{\top},...,\eta_{p_{P}}^{\top},\eta_{c_{1}}^{\top},...,\eta_{c_{C}}^{\top}\right]$
and $\mu:=\left[\mu_{p_{1}}^{\top},...,\mu_{p_{P}}^{\top},\mu_{c_{1}}^{\top},...,\mu_{c_{C}}^{\top},\mu_{M}^{\top}\right]$. 
\end{itemize}
\begin{lemma}\label{prop:pos_semi_def}$Q\succeq0$ for all $x$
that satisfy the market clearing constraint, i.e. $\sum_{p\in P}V_{p}+\sum_{c\in C}V_{c}=0$.
\end{lemma}

\begin{proof}We have 
\[
\begin{array}{rcl}
x^{\top}Qx & = & v^{\top}\left[\begin{array}{cccccc}
\lambda_{p_{1}}Q_{p_{1}} & 0\\
0 & \ddots & 0\\
 & 0 & \lambda_{p_{P}}Q_{p_{P}} & 0\\
 &  & 0 & \lambda_{c_{1}}Q_{c_{1}} & 0\\
 &  &  & 0 & \ddots & 0\\
 &  &  &  & 0 & \lambda_{c_{C}}Q_{c_{C}}
\end{array}\right]v\\
\\
 &  & +2\mathbb{E}^{\mathbb{P}}\left[\Pi\right]^{\top}\left(\sum_{c\in C}V_{c}+\sum_{p\in P}V_{p}\right)\\
\\
 & = & \sum_{p\in P}\lambda_{p}v_{p}^{\top}Q_{p}v_{p}+\sum_{c\in C}\lambda_{c}v_{c}^{\top}Q_{c}v_{c}\\
\\
 & \geq & 0,
\end{array}
\]
since $\lambda_{p}>0$, $\lambda_{c}>0$, $Q_{c}\succ0$, $Q_{p}\succeq0$
for all $p\in P$ and $c\in C$. \end{proof}

In this setting we can reformulate the KKT conditions (\ref{eq:KKT_all-2})
as follows,
\begin{equation}
\begin{array}{rcl}
-\pi-Qx-B^{\top}\eta-A^{\top}\mu & = & 0\\
\\
\eta^{\top}\left(Bx-b\right) & = & 0\\
\\
Bx-b & \leq & 0\\
\\
Ax-a & = & 0\\
\\
\eta & \geq & 0\\
\\
\mu_{M} & = & 0.
\end{array}\label{eq:KKT_all-3}
\end{equation}

\begin{proposition}\label{prop:KKT-cond}The KKT conditions (\ref{eq:KKT_all-3})
are equivalent to the KKT conditions (\ref{eq:KKT_all-2}). \end{proposition}

\begin{proof}We will consider each equation separately. Let us start
with $Ax-a=0$. Writing the equation in matrix form
\[
\left[\begin{array}{ccccccc}
A_{p_{1}} & 0 &  &  &  &  & 0\\
0 & \ddots & 0 &  &  &  & \vdots\\
 & 0 & A_{p_{P}} & 0 &  &  & 0\\
 &  & 0 & A_{c_{1}} & 0 &  & 0\\
 &  &  & 0 & \ddots & 0 & \vdots\\
 &  &  &  & 0 & A_{c_{C}} & 0\\
M_{p_{1}} & \cdots & M_{p_{P}} & I & \cdots & I & 0
\end{array}\right]\left[\begin{array}{c}
v_{p_{1}}\\
\vdots\\
v_{p_{P}}\\
v_{c_{1}}\\
\vdots\\
v_{c_{C}}\\
\mathbb{E}^{\mathbb{P}}\left[\Pi\right]
\end{array}\right]-\left[\begin{array}{c}
a_{p_{1}}\\
\vdots\\
a_{p_{P}}\\
a_{c_{1}}\\
\vdots\\
a_{c_{C}}\\
0
\end{array}\right]=\left[\begin{array}{c}
0\\
\vdots\\
0\\
0\\
\vdots\\
0\\
0
\end{array}\right]
\]
leads to
\[
\left[\begin{array}{c}
A_{p_{1}}v_{p_{1}}\\
\vdots\\
A_{p_{P}}v_{p_{P}}\\
A_{c_{1}}v_{c_{1}}\\
\vdots\\
A_{c_{C}}v_{c_{C}}\\
\sum_{p\in P}V_{p}+\sum_{c\in C}V_{c}
\end{array}\right]-\left[\begin{array}{c}
a_{p_{1}}\\
\vdots\\
a_{p_{P}}\\
a_{c_{1}}\\
\vdots\\
a_{c_{C}}\\
0
\end{array}\right]=\left[\begin{array}{c}
0\\
\vdots\\
0\\
0\\
\vdots\\
0\\
0
\end{array}\right],
\]
as required.

We continue with the expression $-\pi-Qx-B^{\top}\eta-A^{\top}\mu=0$.
We first focus on $-\pi-Qx$ only. Writing the equation in matrix
form 
\[
\begin{array}{l}
-\left[\begin{array}{c}
\mathbb{E}^{\mathbb{P}}\left[\pi_{0,p_{1}}\right]\\
\vdots\\
\mathbb{E}^{\mathbb{P}}\left[\pi_{0,p_{P}}\right]\\
0\\
\vdots\\
0\\
0
\end{array}\right]-\left[\begin{array}{ccccccc}
\lambda_{p_{1}}Q_{p_{1}} & 0 &  &  &  &  & M_{p_{1}}^{\top}\\
0 & \ddots & 0 &  &  &  & \vdots\\
 & 0 & \lambda_{p_{P}}Q_{p_{P}} & 0 &  &  & M_{p_{P}}^{\top}\\
 &  & 0 & \lambda_{c_{1}}Q_{c_{1}} & 0 &  & I\\
 &  &  & 0 & \ddots & 0 & \vdots\\
 &  &  &  & 0 & \lambda_{c_{C}}Q_{c_{C}} & I\\
M_{p_{1}} & \cdots & M_{p_{P}} & I & \cdots & I & 0
\end{array}\right]\left[\begin{array}{c}
v_{p_{1}}\\
\vdots\\
v_{p_{P}}\\
v_{c_{1}}\\
\vdots\\
v_{c_{C}}\\
\mathbb{E}^{\mathbb{P}}\left[\Pi\right]
\end{array}\right]\\
\\
=-\left[\begin{array}{c}
\mathbb{E}^{\mathbb{P}}\left[\pi_{0,p_{1}}\right]\\
\vdots\\
\mathbb{E}^{\mathbb{P}}\left[\pi_{0,p_{P}}\right]\\
0\\
\vdots\\
0\\
0
\end{array}\right]-\left[\begin{array}{c}
\lambda_{p_{1}}Q_{p_{1}}v_{p_{1}}+M_{p_{1}}^{\top}\mathbb{E}^{\mathbb{P}}\left[\Pi\right]\\
\vdots\\
\lambda_{p_{P}}Q_{p_{P}}v_{p_{P}}+M_{p_{P}}^{\top}\mathbb{E}^{\mathbb{P}}\left[\Pi\right]\\
\lambda_{c_{1}}Q_{c_{1}}v_{c_{1}}+\mathbb{E}^{\mathbb{P}}\left[\Pi\right]\\
\vdots\\
\lambda_{c_{C}}Q_{c_{C}}v_{c_{C}}+\mathbb{E}^{\mathbb{P}}\left[\Pi\right]\\
\sum_{p\in P}V_{p}+\sum_{c\in C}V_{c}
\end{array}\right],
\end{array}
\]
and noting that $-\mathbb{E}^{\mathbb{P}}\left[\pi_{0,p}\right]-M_{p}^{\top}\mathbb{E}^{\mathbb{P}}\left[\Pi\right]=-\mathbb{E}^{\mathbb{P}}\left[\pi_{p}\right]$
for all $p\in P$, this leads to
\[
-\pi-Qx=-\left[\begin{array}{c}
\lambda_{p_{1}}Q_{p_{1}}v_{p_{1}}+\mathbb{E}^{\mathbb{P}}\left[\pi_{p_{1}}\right]\\
\vdots\\
\lambda_{p_{P}}Q_{p_{P}}v_{p_{P}}+\mathbb{E}^{\mathbb{P}}\left[\pi_{p_{P}}\right]\\
\lambda_{c_{1}}Q_{c_{1}}v_{c_{1}}+\mathbb{E}^{\mathbb{P}}\left[\Pi\right]\\
\vdots\\
\lambda_{c_{C}}Q_{c_{C}}v_{c_{C}}+\mathbb{E}^{\mathbb{P}}\left[\Pi\right]\\
\sum_{p\in P}V_{p}+\sum_{c\in C}V_{c}
\end{array}\right].
\]
Similarly, writing the remainder $-B^{\top}\eta-A^{\top}\mu$ in matrix
form leads to
\[
-\left[\begin{array}{cccccc}
B_{p_{1}}^{\top} & 0\\
0 & \ddots & 0\\
 & 0 & B_{p_{P}}^{\top} & 0\\
 &  & 0 & B_{c_{1}}^{\top} & 0\\
 &  &  & 0 & \ddots & 0\\
 &  &  &  & 0 & B_{c_{C}}^{\top}\\
0 &  & 0 & 0 &  & 0
\end{array}\right]\left[\begin{array}{c}
\eta_{p_{1}}\\
\vdots\\
\eta_{p_{P}}\\
\eta_{c_{1}}\\
\vdots\\
\eta_{c_{C}}
\end{array}\right]=-\left[\begin{array}{c}
B_{p_{1}}^{\top}\eta_{p_{1}}\\
\vdots\\
B_{p_{P}}^{\top}\eta_{p_{P}}\\
B_{c_{1}}^{\top}\eta_{c_{1}}\\
\vdots\\
B_{c_{C}}^{\top}\eta_{c_{C}}\\
0
\end{array}\right]
\]
and 
\[
-\left[\begin{array}{ccccccc}
A_{p_{1}}^{\top} & 0 &  &  &  &  & M_{p}^{\top}\\
0 & \ddots & 0 &  &  &  & \vdots\\
 & 0 & A_{p_{P}}^{\top} & 0 &  &  & M_{p}^{\top}\\
 &  & 0 & A_{c_{1}}^{\top} & 0 &  & I\\
 &  &  & 0 & \ddots & 0 & \vdots\\
 &  &  &  & 0 & A_{c_{C}}^{\top} & I\\
0 & \cdots & 0 & 0 & \cdots & 0 & 0
\end{array}\right]\left[\begin{array}{c}
\mu_{p_{1}}\\
\vdots\\
\mu_{p_{P}}\\
\mu_{c_{1}}\\
\vdots\\
\mu_{c_{C}}\\
\mu_{M}
\end{array}\right]=-\left[\begin{array}{c}
A_{p_{1}}^{\top}\mu_{p_{1}}\\
\vdots\\
A_{p_{P}}^{\top}\mu_{p_{P}}\\
A_{c_{1}}^{\top}\mu_{c_{1}}\\
\vdots\\
A_{c_{C}}^{\top}\mu_{c_{C}}\\
0
\end{array}\right]
\]
where $\mu_{M}=0$ was used. Writing it all together 
\[
-\left[\begin{array}{c}
\lambda_{p_{1}}Q_{p_{1}}v_{p_{1}}+\mathbb{E}^{\mathbb{P}}\left[\pi_{p_{1}}\right]\\
\vdots\\
\lambda_{p_{P}}Q_{p_{P}}v_{p_{P}}+\mathbb{E}^{\mathbb{P}}\left[\pi_{p_{P}}\right]\\
\lambda_{c_{1}}Q_{C}v_{c_{1}}+\mathbb{E}^{\mathbb{P}}\left[\Pi\right]\\
\vdots\\
\lambda_{c_{C}}Q_{C}v_{c_{C}}+\mathbb{E}^{\mathbb{P}}\left[\Pi\right]\\
\sum_{p\in P}V_{p}+\sum_{c\in C}V_{c}
\end{array}\right]-\left[\begin{array}{c}
B_{p_{1}}^{\top}\eta_{p_{1}}\\
\vdots\\
B_{p_{P}}^{\top}\eta_{p_{P}}\\
B_{c_{1}}^{\top}\eta_{c_{1}}\\
\vdots\\
B_{c_{C}}^{\top}\eta_{c_{C}}\\
0
\end{array}\right]-\left[\begin{array}{c}
A_{p_{1}}^{\top}\mu_{p_{1}}\\
\vdots\\
A_{p_{P}}^{\top}\mu_{p_{P}}\\
A_{c_{1}}^{\top}\mu_{c_{1}}\\
\vdots\\
A_{c_{C}}^{\top}\mu_{c_{C}}\\
0
\end{array}\right]=\left[\begin{array}{c}
0\\
\vdots\\
0\\
0\\
\vdots\\
0\\
0
\end{array}\right]
\]
gives the required conditions. 

The proofs for $\eta^{\top}\left(Bx-b\right)=0$, $Bx-b\leq0$, and
$\eta\geq0$ are trivial.\end{proof}

Since the additional constraints $\mu_{M}=0$ on the dual variables
of Problem \ref{eq:KKT_all-3} cannot be handled by most of the available
quadratic programming solvers, we have to reformulate the problem
in a dual form. We start by formulating the optimization problem out
of the KKT conditions (\ref{eq:KKT_all-3}) as 
\begin{equation}
\begin{array}{cl}
\underset{x}{\text{max}} & -\pi^{\top}x-\frac{1}{2}x^{\top}Qx\\
\\
\text{s.t.} & Ax=a\\
\\
 & Bx\leq b\\
\\
 & \mu_{M}=0
\end{array}\label{eq:op_all-2}
\end{equation}
and by defining the Lagrangian as

\[
\mathcal{L}\left(x,\mu,\eta\right)=\left\{ \begin{array}{rl}
-\frac{1}{2}x^{\top}Qx-\pi^{\top}x-\left(Ax-a\right)^{\top}\mu-\left(Bx-b\right)^{\top}\eta; & \text{if }\eta\geq0\\
\\
-\infty; & \text{otherwise}.
\end{array}\right.
\]
By the virtue of Lemma \ref{prop:pos_semi_def}, $Q\succeq0$, and
$\mathcal{L}\left(x,\mu,\eta\right)$ is therefore a smooth and convex
function. The unconstrained minimizer can be determined by solving
$\mathcal{D}_{x}\mathcal{L}\left(x,\mu,\eta\right)=0$. Calculating
\[
\mathcal{D}_{x}\mathcal{L}\left(x,\mu,\eta\right)=-Qx-\pi-A^{\top}\mu-B^{\top}\eta
\]
and inserting $\pi$ back to the Lagrangian, an equivalent formulation
is obtained as follows,
\[
\mathcal{L}\left(x,\mu,\eta\right)=\left\{ \begin{array}{rl}
\frac{1}{2}x^{\top}Qx+a^{\top}\mu+b^{\top}\eta & \text{if }\eta\geq0\text{ \text{and }}-Qx-\pi-A^{\top}\mu-B^{\top}\eta=0,\\
\\
-\infty & \text{otherwise}
\end{array}\right.
\]
Relating the latter to a maximization optimization problem, the following
formulation is obtained 
\begin{equation}
\begin{array}{cl}
\underset{x,\mu,\eta}{\text{max}} & -\frac{1}{2}x^{\top}Qx-\mu^{\top}a-\eta^{\top}b\\
\\
\text{s.t.} & Qx+A^{\top}\mu+B^{\top}\lambda+\pi=0\\
\\
 & \eta\geq0\\
\\
 & \mu_{M}=0.
\end{array}\label{eq:op_all-1-1}
\end{equation}
Problem (\ref{eq:op_all-1-1}) is equivalent to Problem (\ref{eq:op_all-2}),
but it can be solved using any quadratic programming algorithm. The
numerical results in this paper are calculated using Gurobi \cite{gurobioptimization2014gurobioptimizer}.

\section{Extensions of the model\label{sec:Extensions}}

In this section we describe a few realistic extensions of the model
from Section \ref{sec:Problem-description}, that are needed to make
the model applicable in practice. Note that all extensions can be
incorporated into the quadratic programming framework, and thus the
reformulation of Section \ref{sec:Algorithm} still applies.

\subsection{Future contracts\label{sub:Future-contracts}}

In the model above, we assumed that the market participants trade
only forward contracts. A more realistic market setting would include
also future contracts. In this section we explain how the model can
be extended in this regard. 

Let $\acute{\Pi}$ and $\grave{\Pi}$ denote prices of forward and
future contracts, respectively. A no-arbitrage argument allows one
to calculate $\mathbb{E}^{\mathbb{P}}\left[\grave{\Pi}\right]$ as
a solution of the following triangular system of linear equations
for each $j\in J$, 
\begin{equation}
\begin{array}{l}
\mathbb{E}^{\mathbb{P}}\left[\acute{\Pi}\left(t_{i},T_{j}\right)\right]=\\
\\
\begin{cases}
\mathbb{E}^{\mathbb{P}}\left[\grave{\Pi}\left(t_{i},T_{j}\right)\right] & i=\max\left\{ I_{j}\right\} \\
\\
\mathbb{E}^{\mathbb{P}}\left[\grave{\Pi}\left(t_{i},T_{j}\right)\right]+\sum_{k=i}^{\max\left\{ I_{j}\right\} -1}\left(\left[\mathbb{E}^{\mathbb{P}}\left[\grave{\Pi}\left(t_{k+1},T_{j}\right)\right]-\mathbb{E}^{\mathbb{P}}\left[\grave{\Pi}\left(t_{k},T_{j}\right)\right]\right]\frac{e^{-rt_{k+1}}}{e^{-rT_{j}}}\right) & i<\max\left\{ I_{j}\right\} .
\end{cases}
\end{array}\label{eq:bi}
\end{equation}
Recall from the introduction that $t_{\max\left\{ I_{j}\right\} }=T_{j}$.

\subsection{Tradable contracts}

In the previous sections, we have implicitly assumed that each contract
covers only one delivery period. As described in Subsection \ref{sec:Organization-of-market},
this is clearly not true in reality where contracts often cover a
longer delivery period. For example, season ahead contracts cover
the delivery over every half hour in the next season. Similarly, a
month ahead contract covers every half hour in the next month. One
can also distinguish between peak contracts which cover only delivery
periods from 7am to 7pm and off-peak contracts which cover delivery
periods from 7pm to 7am. Other blocks are also possible.

A contract that covers many delivery periods $J'\subseteq J$ and
is traded at $t_{i}$, $i\in\cap_{j'\in J'}I_{j'}$ can be incorporated
in our model by enforcing
\begin{equation}
\mathbb{E}^{\mathbb{P}}\left[\Pi\left(t_{i},T_{j'}\right)\right]=\mathbb{E}^{\mathbb{P}}\left[\Pi\left(t_{i},T_{j''}\right)\right]
\end{equation}
and 
\begin{equation}
V_{k}\left(t_{i},T_{j'}\right)=V_{k}\left(t_{i},T_{j''}\right)
\end{equation}
for all $k\in P\cup C$ and all $j',j''\in J'$.

\subsection{Costs of trading\label{sub:Costs-of-trading}}

In the previous sections we assumed that the players can change their
position without paying any fees. The reality is of course different.
Cost of trading is an important component involved in the trading
of power. The precise formulation of the trading costs depends on
the market micro-structure and it is a research area itself. In this
paper we adapt the view of \cite{almgren2001optimal} where they propose
that a cost of trading one unit of power can be approximated by a
linear function 
\begin{equation}
h\left(V_{p}\left(t_{i},T_{j}\right)\right)=\epsilon_{ij}\:\text{sign}\left(V_{p}\left(t_{i},T_{j}\right)\right)+\upsilon_{ij}V_{p}\left(t_{i},T_{j}\right).\label{eq:linear_costs}
\end{equation}
The first term represents the fixed costs of trading that occur due
to the bid-ask spread and trading fees. A good estimate for $\epsilon_{ij}$
is the sum of a half of the bid-offer spread and the trading fees.
The second term approximates the micro-structure of the order book.
Selling a large volume $V_{p}\left(t_{i},T_{j}\right)$ exhausts the
supply of liquidity, which causes a short-term decrease in the price.
The factor $\upsilon_{ij}$ hence represents a half of the change
in price caused by selling volume $V_{p}\left(t_{i},T_{j}\right)$.
We assume that the decrease in price is only temporary and that the
price returns the equilibrium level at the next trading time $t_{i+1}$. 

Equation (\ref{eq:linear_costs}) represents costs of trading a unit
of power. Costs of trading a volume $V_{p}\left(t_{i},T_{j}\right)$
is then
\begin{equation}
V_{p}\left(t_{i},T_{j}\right)h\left(V_{p}\left(t_{i},T_{j}\right)\right)=\epsilon_{ij}\left|V_{p}\left(t_{i},T_{j}\right)\right|+\upsilon_{ij}V_{p}\left(t_{i},T_{j}\right)^{2}.\label{eq:linear_costs-1}
\end{equation}
See Subsection \ref{sec:Organization-of-market} for a discussion
of the bid-ask spread in the UK electricity market.

\section{Numerical results\label{sec:Simulations}}

In this section we describe numerical results. First we investigate
the calibration of power plant parameters from the historical production.
Then we discuss the properties of the term structure of electricity
with the help of a simple artificial example. Finally, the suitability
of our model for realistically modeling the entire UK electricity
market is presented.

\subsection{Calibration\label{sub:Calibration}}

If our model is applied in practice, one has to estimate the physical
characteristics of the power plants such as capacity, ramp up and
ramp down constraints, efficiency, and carbon emission intensity factor.

In the UK all power plants are required to submit their available
capacity as well as ramp up and ramp down constraints to the system
operator on a half hourly basis. This data is publicly available at
the Elexon website%
\footnote{http://www.bmreports.com/%
}.

A more challenging problem is to estimate the efficiency and the carbon
emission intensity factor of each power plant. As we explained in
Section \ref{sec:Organization-of-market}, all market participants
submit their expected production/consumption to the system operator
one hour before the delivery. Our goal is to find the efficiency and
the carbon emission intensity factor that best describe the historical
production of each power plant. 

We first normalized the historical production of each power plant
by calculating the ratio between the production and the available
capacity of a power plant in each half hour. We denote the normalized
historical production by $\tilde{W}_{p,l,r}\left(T_{j}\right)$ for
each delivery period $j\in J$ and power plant $r\in R^{p,l}$. Normalization
makes sure that $\tilde{W}_{p,l,r}\left(T_{j}\right)\in\left[0,1\right]$.
For the purpose of calibration we assume that producers are risk neutral
and set $\lambda_{p}=0$ for all $p\in P$. Furthermore, we neglect
the ramp up and ramp down constraints (\ref{eq:pb0}). With these
simplifications, we can see that a power plant will produce at time
$T_{j}$ if and only if the income from selling electricity at the
spot price is greater or equal to the costs of purchasing the required
fuel and emission certificates at the current spot price. In other
words, for a power plant that runs on fuel $l\in L$ and produces
electricity at time $T_{j}$,
\begin{equation}
\Pi\left(T_{j},T_{j}\right)-c^{p,l,r}G_{l}\left(T_{j},T_{j}\right)-g^{p,l,r}G_{em}\left(T_{j},T_{j}\right)\geq0\label{eq:e2}
\end{equation}
must hold for production to take place.

It is immediately clear why (\ref{eq:e2}) must hold when only spot
contracts are available. Let us investigate why (\ref{eq:e2}) holds
also if forward and future electricity contracts are available on
the market. At any trading time $t_{i}$, $i\in I_{j}$, a rational
producer could enter into a short electricity forward contract and
simultaneously into a long fuel and emission forward contract if
\begin{equation}
\Pi\left(t_{i},T_{j}\right)-c^{p,l,r}G_{l}\left(t_{i},T_{j}\right)-g^{p,l,r}G_{em}\left(t_{i},T_{j}\right)\geq0.\label{eq:prodIneq-1}
\end{equation}
At delivery time $T_{j}$, this producer has two options:
\begin{itemize}
\item To acquire the delivery of the fuel and emission certificates bought
at trading time $t_{i}$ and produce electricity. In this case, she
observes the following profit 
\begin{equation}
\widehat{P_{1}}\left(T_{j}\right)=\Pi\left(t_{i},T_{j}\right)-c^{p,l,r}G_{l}\left(t_{i},T_{j}\right)-g^{p,l,r}G_{em}\left(t_{i},T_{j}\right).
\end{equation}

\item To produce no electricity and instead close the forward electricity,
fuel, and emission contracts. In this case, she observes the following
profit 
\begin{equation}
\begin{array}{rcl}
\widehat{P_{2}}\left(T_{j}\right) & = & \left[\Pi\left(t_{i},T_{j}\right)-\Pi\left(T_{j},T_{j}\right)\right]-c^{p,l,r}\left[G_{l}\left(t_{i},T_{j}\right)-G_{l}\left(T_{j},T_{j}\right)\right]\\
\\
 &  & -g^{p,l,r}\left[G_{em}\left(t_{i},T_{j}\right)-G_{em}\left(T_{j},T_{j}\right)\right].
\end{array}
\end{equation}

\end{itemize}
Power plant $r\in R^{p,l}$ will run at $T_{j}$ if and only if
\begin{equation}
\widehat{P_{1}}\left(T_{j}\right)\geq\widehat{P_{2}}\left(T_{j}\right).\label{eq:e1}
\end{equation}
With some reordering of the terms, it is easy to see that inequality
(\ref{eq:e1}) is equivalent to inequality (\ref{eq:e2}).

To adjust inequality (\ref{eq:e2}) for the neglected risk premium
and trading costs, we add another term $\tilde{c}^{p,l,r}\geq0$ and
then (\ref{eq:e2}) reads
\begin{equation}
\Theta\left(c^{p,l,r},g^{p,l,r},\tilde{c}^{p,l,r};T_{j}\right)\geq0,\label{eq:prodIneq}
\end{equation}
where
\begin{equation}
\Theta\left(c^{p,l,r},g^{p,l,r},\tilde{c}^{p,l,r};T_{j}\right)=\Pi\left(T_{j},T_{j}\right)-c^{p,l,r}G_{l}\left(T_{j},T_{j}\right)-g^{p,l,r}G_{em}\left(T_{j},T_{j}\right)-\tilde{c}^{p,l,r}.\label{eq:mr}
\end{equation}

In last step of the calibration, we use $\Theta\left(c^{p,l,r},g^{p,l,r},\tilde{c}^{p,l,r};T_{j}\right)$
in the logistic regression. The efficiency $c^{p,l,r}$ and carbon
emission intensity factor $g^{p,l,r}$ that best explain the historical
production of power plant $r\in R^{p,l}$ are found as an optimal
solution to 
\begin{equation}
\min_{c^{p,l,r},g^{p,l,r},\tilde{c}^{p,l,r}}\sum_{j\in J}\left(\tilde{W}_{p,l,r}\left(T_{j}\right)-\frac{1}{1+\exp\left(-\Theta\left(c^{p,l,r},g^{p,l,r},\tilde{c}^{p,l,r};T_{j}\right)\right)}\right)^{2}.\label{eq:optPar}
\end{equation}
For each power plant we used over 5000 training samples obtained from
the period between 1/1/2012 and 1/1/2013.

\subsection{Term structure of the price}

In this subsection we study the term structure of the electricity
price. We focus on a simple problem with only one delivery period
and five trading periods. The trading periods named by a number from
1 to 5 correspond to a two months ahead, a month ahead, a week ahead,
a day ahead, and the spot price, respectively. To keep the effects
of various parameters of the model on the term structure of the electricity
price clearly visible, we set the risk free interest rate $\hat{r}$
to zero.

Figure \ref{fig:f1} shows how the term structure of the electricity
price and the number of contracts traded depends on the risk aversion
of the players involved. We can see that the market is usually in
normal backwardation (i.e. the term structure is upward sloping) and
that the slope as well as the price increases when the risk aversion
of the players increases. This increase in the price is due to the
increased risk premium of the producers. The risk premium for risky
contracts is higher than the risk premium of less risky contracts.

\begin{figure}
\begin{centering}
\includegraphics[bb=35bp 180bp 545bp 600bp,clip,scale=0.4]{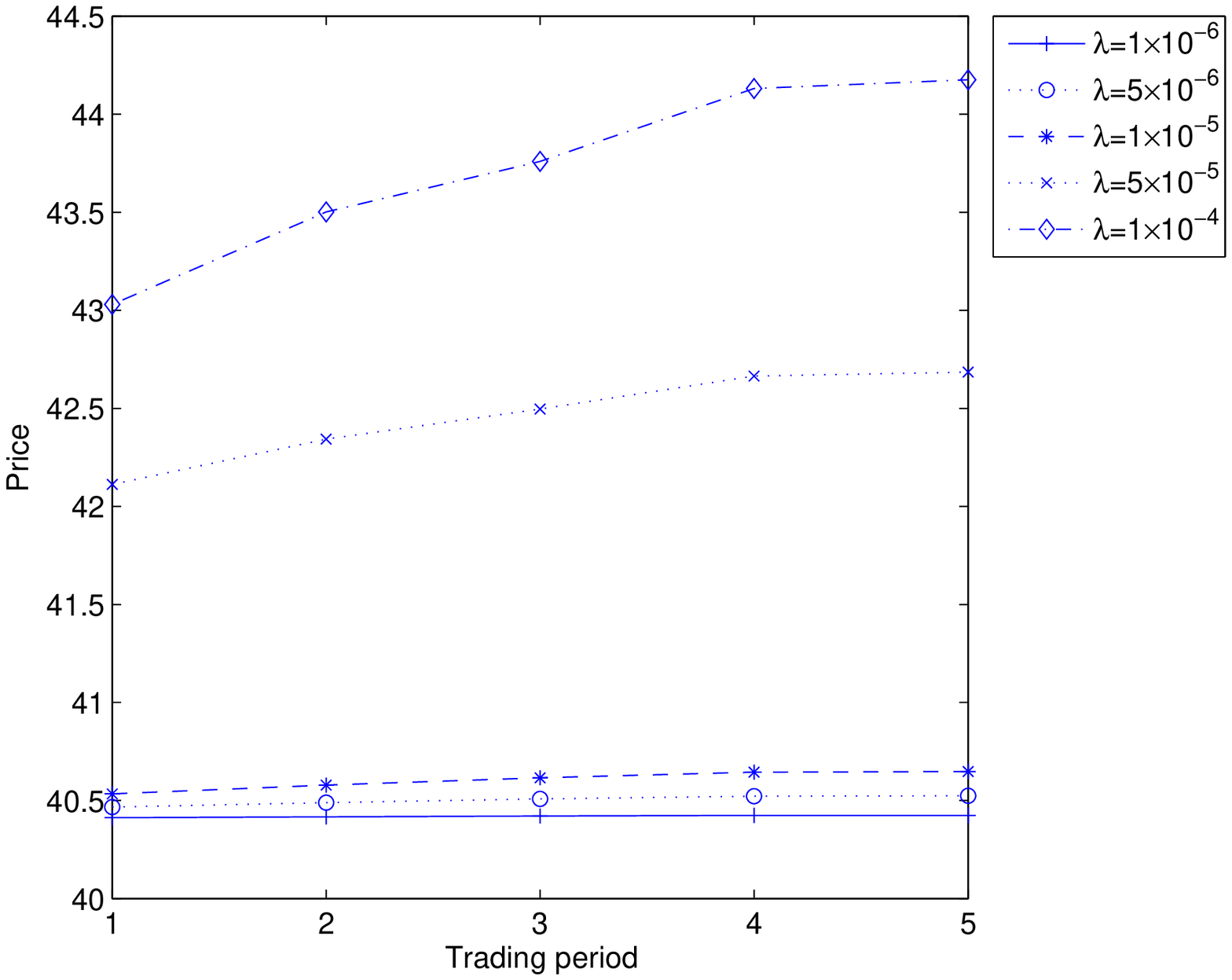}\includegraphics[bb=35bp 180bp 545bp 600bp,clip,scale=0.4]{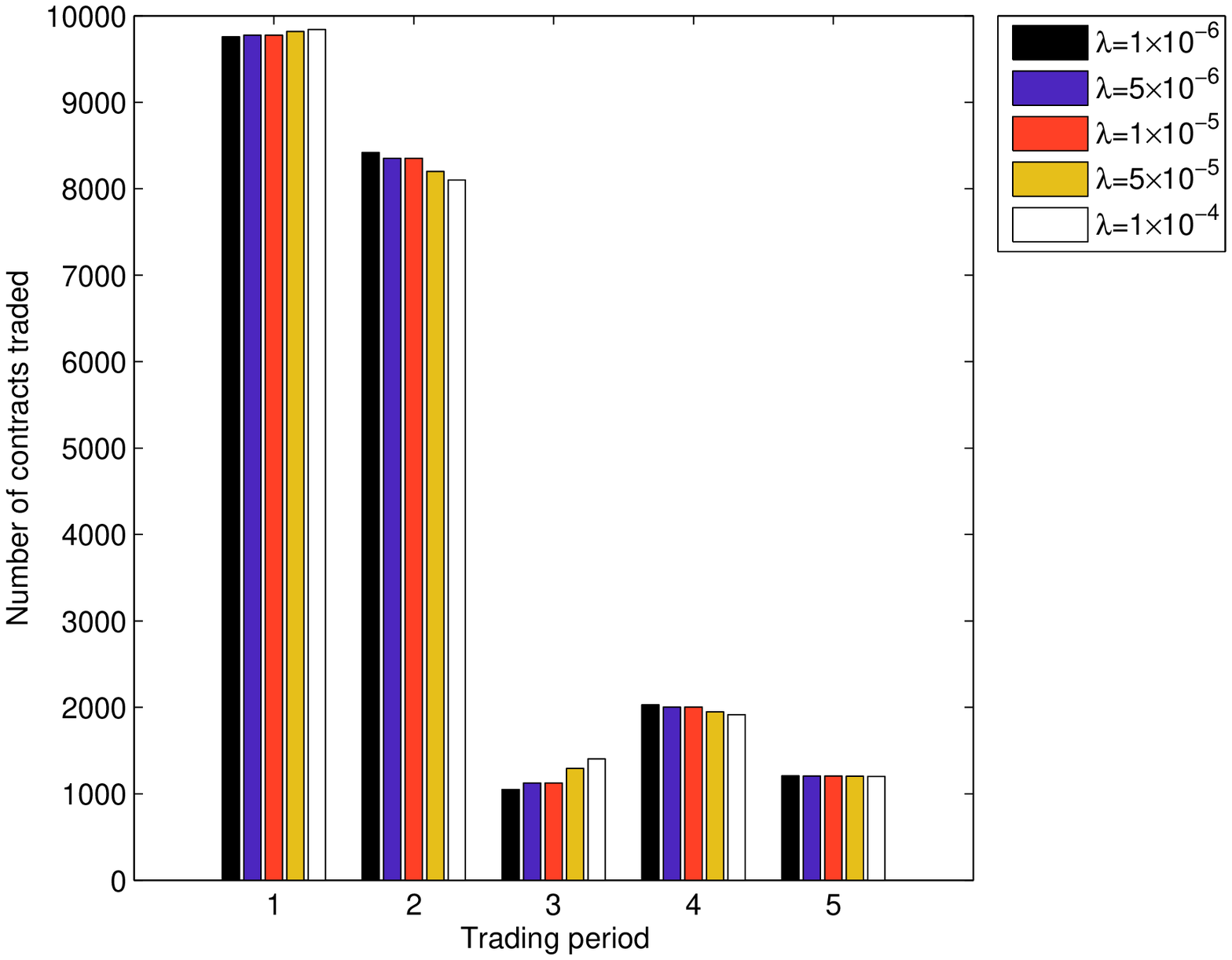}
\par\end{centering}

\caption{\label{fig:f1}$\upsilon_{ij}=0$, $\epsilon_{ij}=0$ for all $i\in\left\{ 1,...,5\right\} $
and $j\in\left\{ 1\right\} $, and increasing $\lambda_{k}$ for all
$k\in P\cup C$. The trading periods named by a number from 1 to 5
correspond to a two months ahead, a month ahead, a week ahead, a day
ahead, and the spot price, respectively.}

\end{figure}

In Figure \ref{fig:f2} we study the case when only the risk aversion
of the producers is increased while the risk aversion of the consumers
is kept constant. As we can see the risk aversion of the producers
does not have any significant impact on the shape of the term structure
of the price. However, the price of all contracts increased significantly
due to the increased risk premium of the producers. 

\begin{figure}
\begin{centering}
\includegraphics[bb=35bp 180bp 545bp 600bp,clip,scale=0.4]{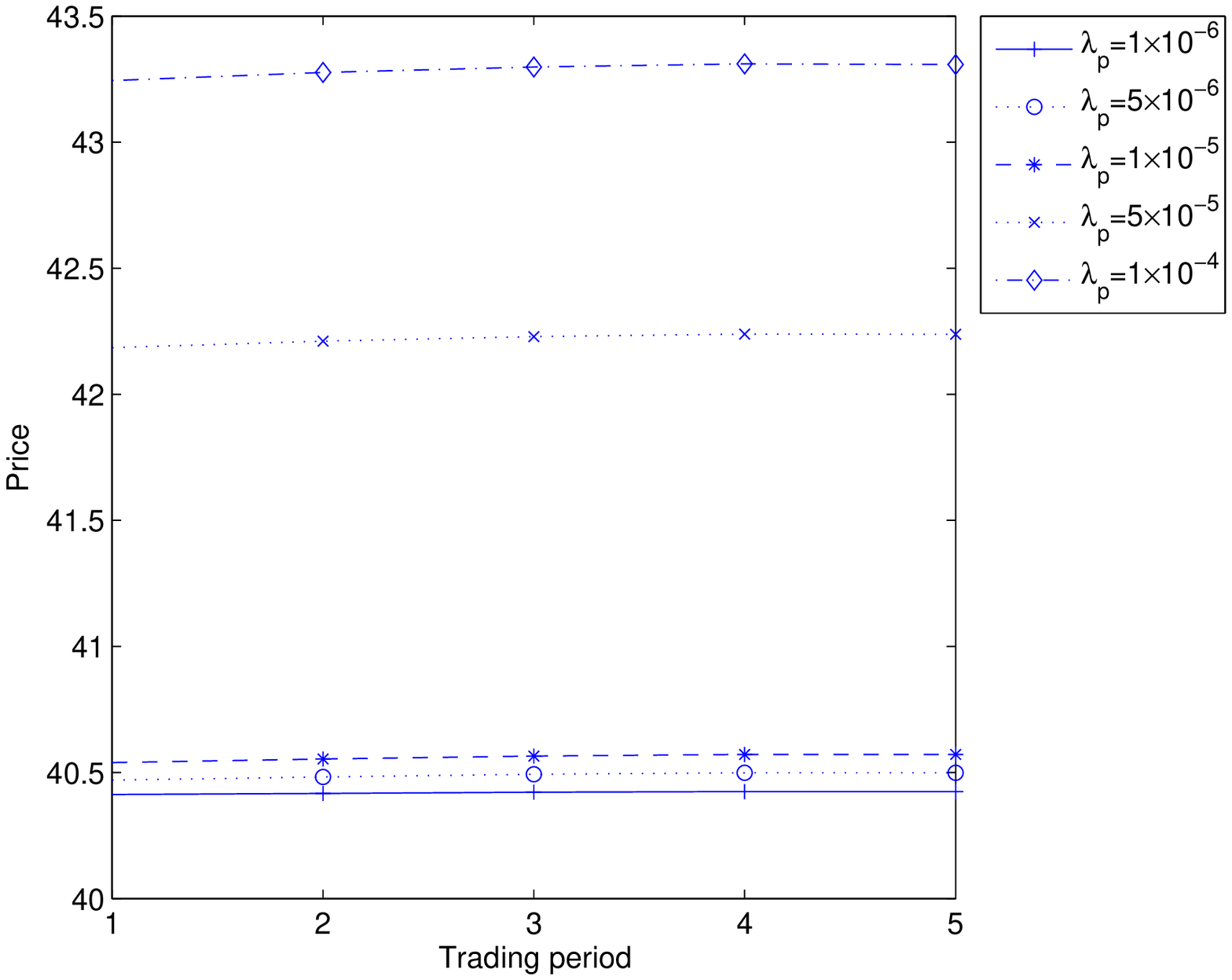}\includegraphics[bb=35bp 180bp 545bp 600bp,clip,scale=0.4]{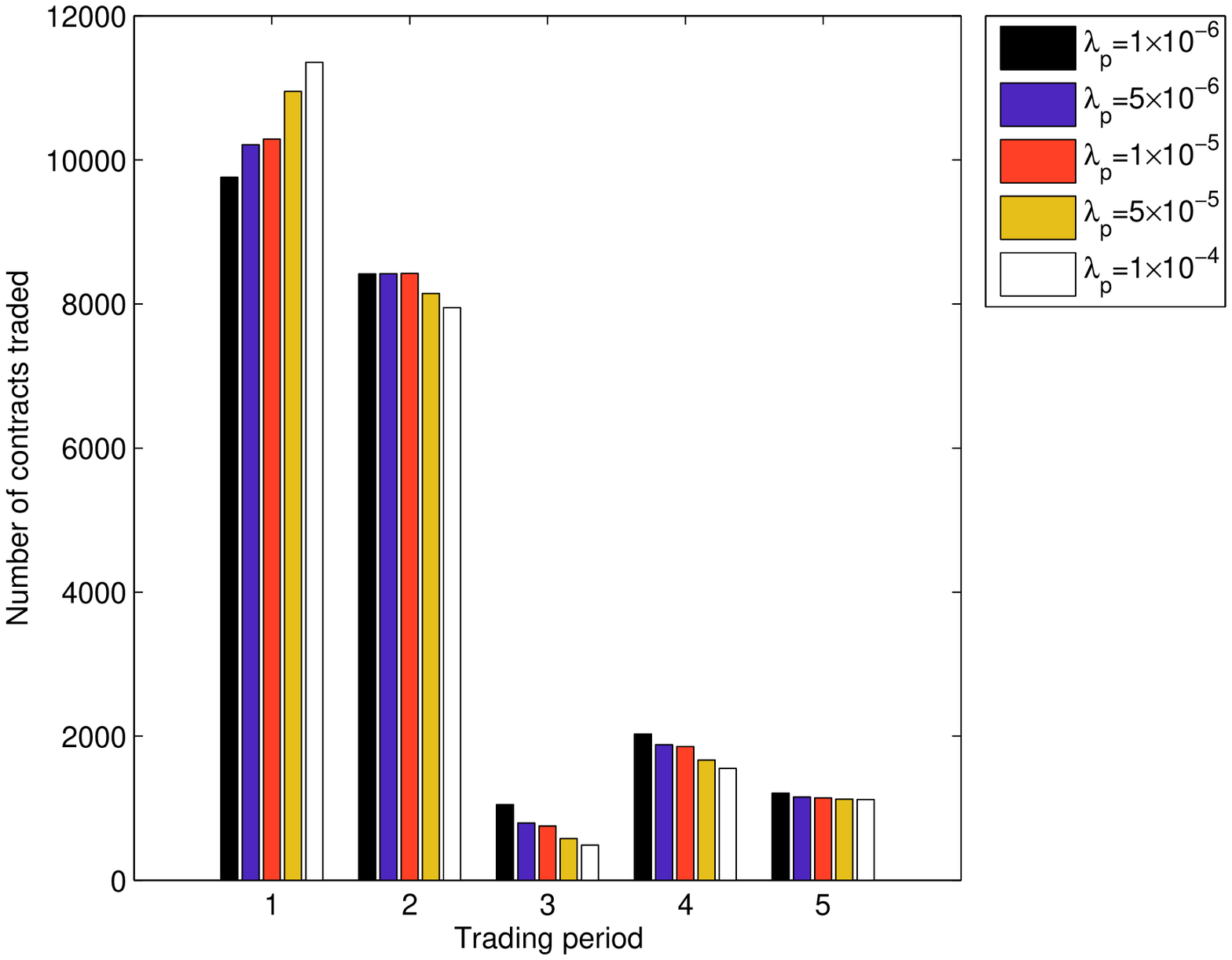}
\par\end{centering}

\caption{\label{fig:f2}$\upsilon_{ij}=0$, $\epsilon_{ij}=0$ for all $i\in\left\{ 1,...,5\right\} $
and $j\in\left\{ 1\right\} $, $\lambda_{c}=10^{-6}$ for all $c\in C$,
and increasing $\lambda_{p}$ for all $p\in P$. }
\end{figure}

In Figure \ref{fig:f3} we study the case when only the risk aversion
of the consumers is increased while the risk aversion of the producers
is kept constant. We can see that this has a small impact on the term
structure of the price. In contrast to the producers, the increased
risk premium of the consumers is not reflected in an increased price.
Thus, the consumers can only maintain their profitability by increasing
the fixed price they charge the end users (i.e. they increase the
constant price $s_{c}$ defined in (\ref{eq:23})).

\begin{figure}
\begin{centering}
\includegraphics[bb=35bp 180bp 545bp 600bp,clip,scale=0.4]{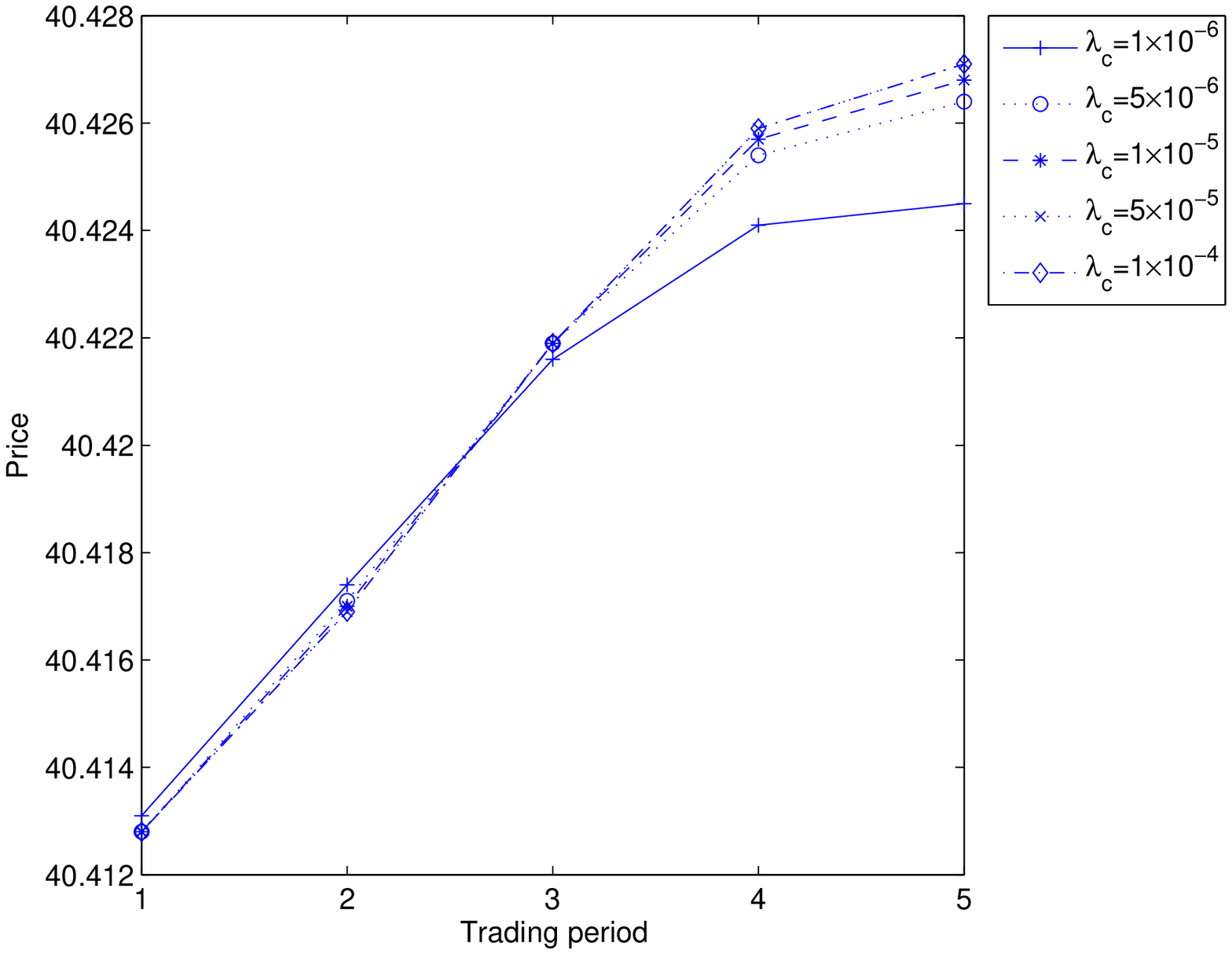}\includegraphics[bb=35bp 180bp 545bp 600bp,clip,scale=0.4]{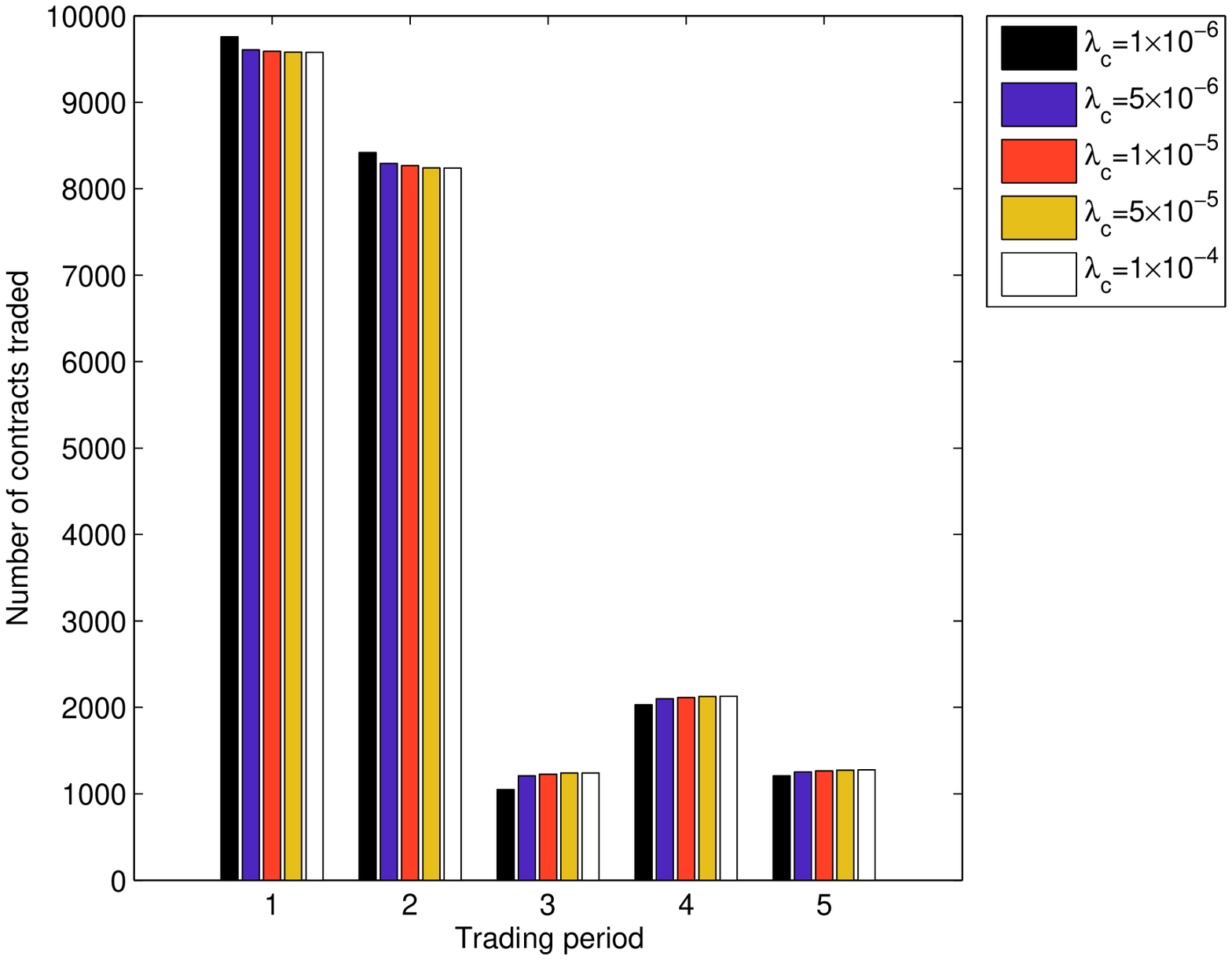}
\par\end{centering}

\caption{\label{fig:f3}$\upsilon_{ij}=0$, $\epsilon_{ij}=0$ for all $i\in\left\{ 1,...,5\right\} $
and $j\in\left\{ 1\right\} $, $\lambda_{p}=10^{-6}$ for all $p\in P$,
and increasing $\lambda_{c}$ for all $c\in C$.}
\end{figure}

One must increase the consumers' risk aversion parameter $\lambda_{c}$
significantly to observe a visual impact. Such hypothetical case is
depicted in Figure \ref{fig:f3-1}. We can see that by increasing
the risk aversion of consumers, the term structure of electricity
prices changes from the normal backwardation to contango. An increased
risk aversion increases the consumers' interest in less risky contracts
and decreases the interest in risky contracts. Since the contracts
closer to delivery are usually more risky, this causes the term structure
to flip from normal backwardation to contango. A small kink at the
fourth trading period is a consequence of a calibration of the covariance
matrix on historical values.

\begin{figure}
\begin{centering}
\includegraphics[bb=35bp 180bp 545bp 600bp,clip,scale=0.4]{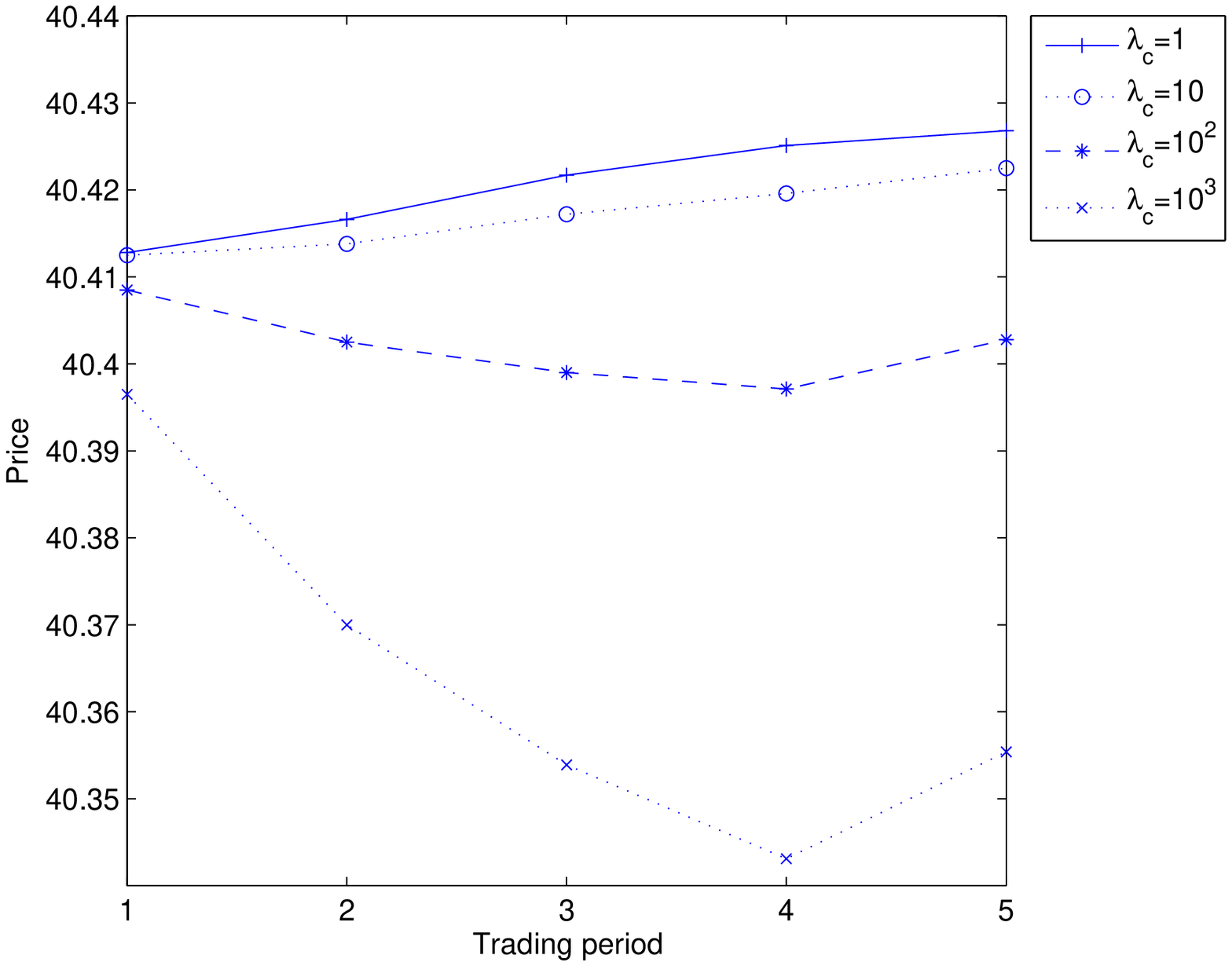}\includegraphics[bb=35bp 180bp 545bp 600bp,clip,scale=0.4]{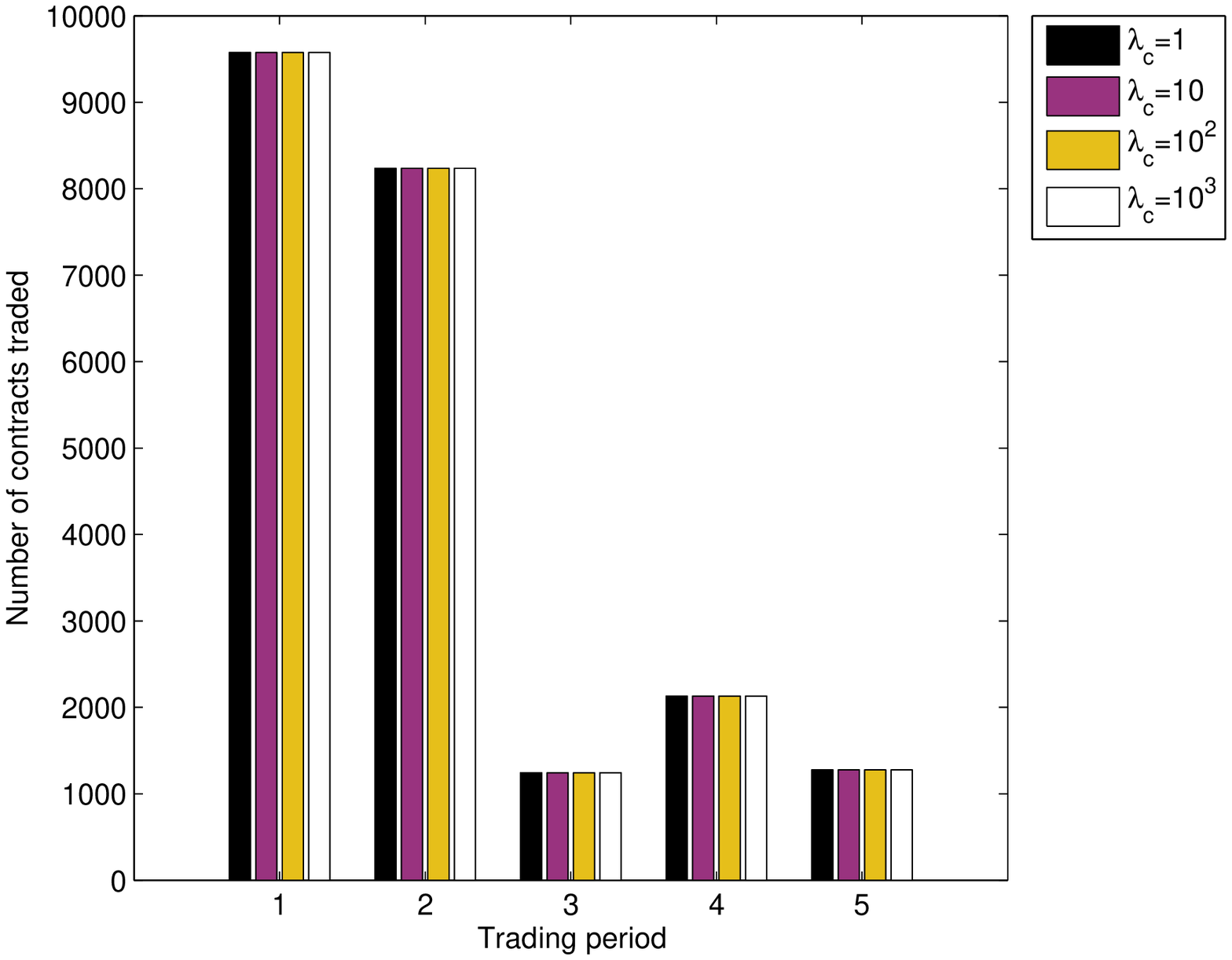}
\par\end{centering}

\begin{centering}
\caption{\label{fig:f3-1}$\upsilon_{ij}=0$, $\epsilon_{ij}=0$ for all $i\in\left\{ 1,...,5\right\} $
and $j\in\left\{ 1\right\} $, $\lambda_{p}=10^{-6}$ for all $p\in P$,
and increasing $\lambda_{c}$ for all $c\in C$.}

\par\end{centering}

\end{figure}

Intuitively, it may not be immediately clear why the change in the
risk aversion of producers and consumers has such a different impact
on the term structure of the electricity price. The level of the price
is defined by producers only. Let us consider a simplified setting
when there is only one trading period and only one delivery period.
Moreover, we have just one producer and one consumer. In such a setting,
the profit of the consumer can be calculated as 
\begin{equation}
P_{c}\left(V_{c},\Pi\right)=e^{-\hat{r}T_{1}}\left(-\Pi\left(t_{1},T_{1}\right)+s_{c}\right)D\left(T_{1}\right)
\end{equation}
where the constraint $V_{c}\left(t_{1},T_{1}\right)=D\left(T_{1}\right)$
was used. We can see that the profit function does not depend on the
consumer's decision variables $V_{c}$. Thus the consumer does not
have any power to influence the price. She has to buy a sufficient
amount of electricity to satisfy the demand of end users regardless
of the electricity price. On the other hand, the producer has much
more flexibility. Her profit can be calculated as
\begin{equation}
P_{p}\left(V_{p},\Pi\right)=e^{-\hat{r}T_{1}}\left[-\Pi\left(t_{1},T_{1}\right)V_{p}\left(t_{1},T_{1}\right)-O_{p}\left(t_{1},T_{1}\right)G_{em}\left(t_{1},T_{1}\right)-\sum_{l\in L}G_{l}\left(t_{1},T_{1}\right)F_{p,l}\left(t_{1},T_{1}\right)\right].
\end{equation}
An obvious feasible solution to her optimization problem (\ref{eq:opt_prod})
is 
\begin{equation}
V_{p}\left(t_{1},T_{1}\right)=O_{p}\left(t_{1},T_{1}\right)=F_{p,l}\left(t_{1},T_{1}\right)=0
\end{equation}
for all $l\in L$. This leads to the objective value $\Phi_{p}=0$.
Thus, clearly for any given price $\Pi\left(t_{1},T_{1}\right)$,
the objective value satisfies $\Phi_{p}\geq0$. A producer will only
turn on the power plants if the electricity price is high enough to
cover all production costs, trading costs, and the risk premium. A
similar reasoning can also be extended to a multiperiod setting.

In Figures \ref{fig:f4}, \ref{fig:f5}, \ref{fig:f6}, and \ref{fig:f7}
we study the impact of the liquidity on the term structure of the
electricity price. In Figures \ref{fig:f4}, and \ref{fig:f5} the
impact of the market micro-structure (i.e. the quadratic term $\upsilon$)
is examined while in Figures \ref{fig:f6}, and \ref{fig:f7} the
impact of the bid-ask spread and the trading fees (i.e. the linear
term $\epsilon$) is depicted. 

Increasing of $\upsilon$ for all the contracts simultaneously (see
Figure \ref{fig:f4}), increases the price without changing the shape
of the term structure. It has a large impact on the optimal trading
strategy of the players. When $\upsilon$ is large players spread
the number of contracts traded in each trading period equally among
all the trading periods available.

\begin{figure}
\begin{centering}
\includegraphics[bb=35bp 180bp 545bp 600bp,clip,scale=0.4]{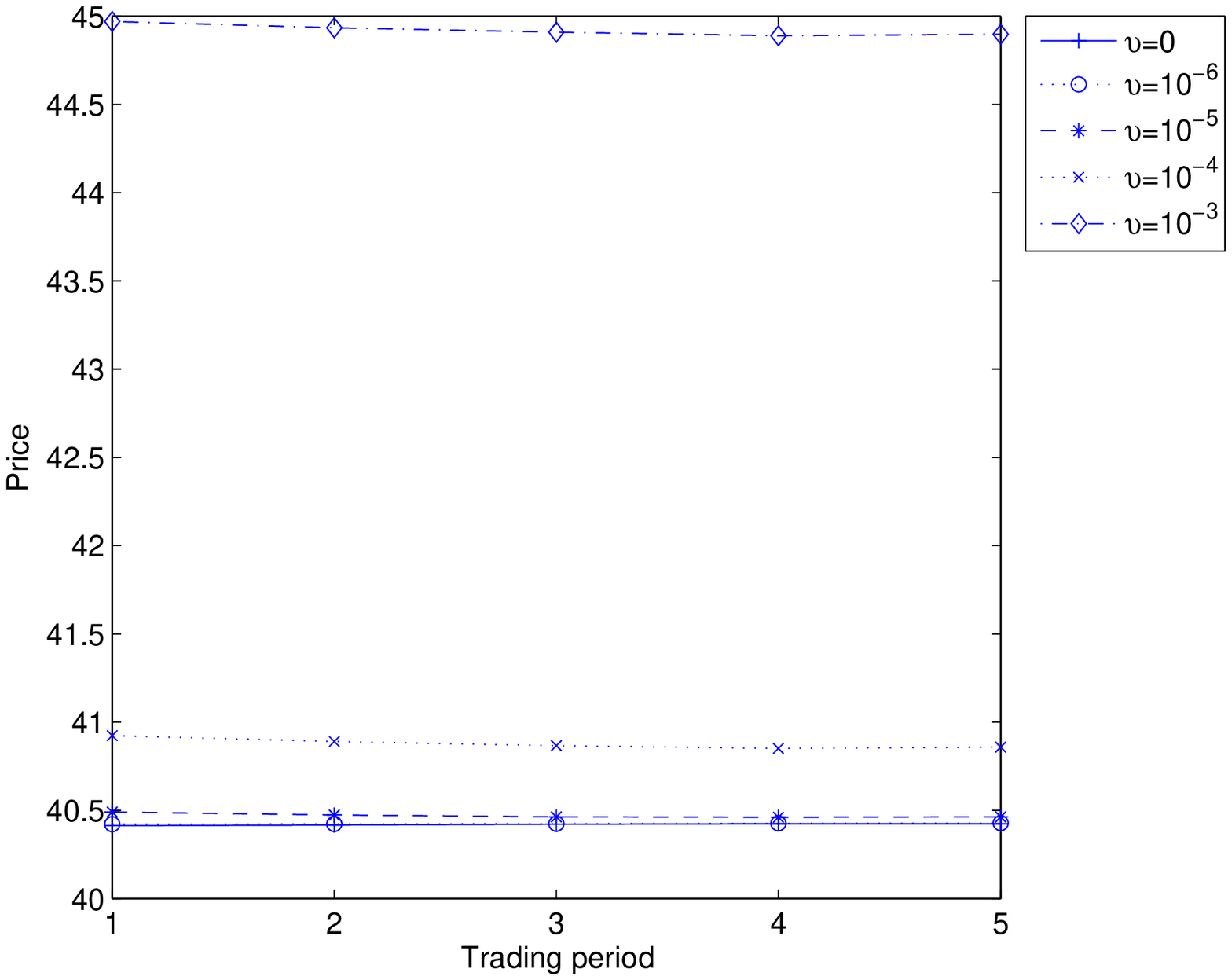}\includegraphics[bb=35bp 180bp 545bp 600bp,clip,scale=0.4]{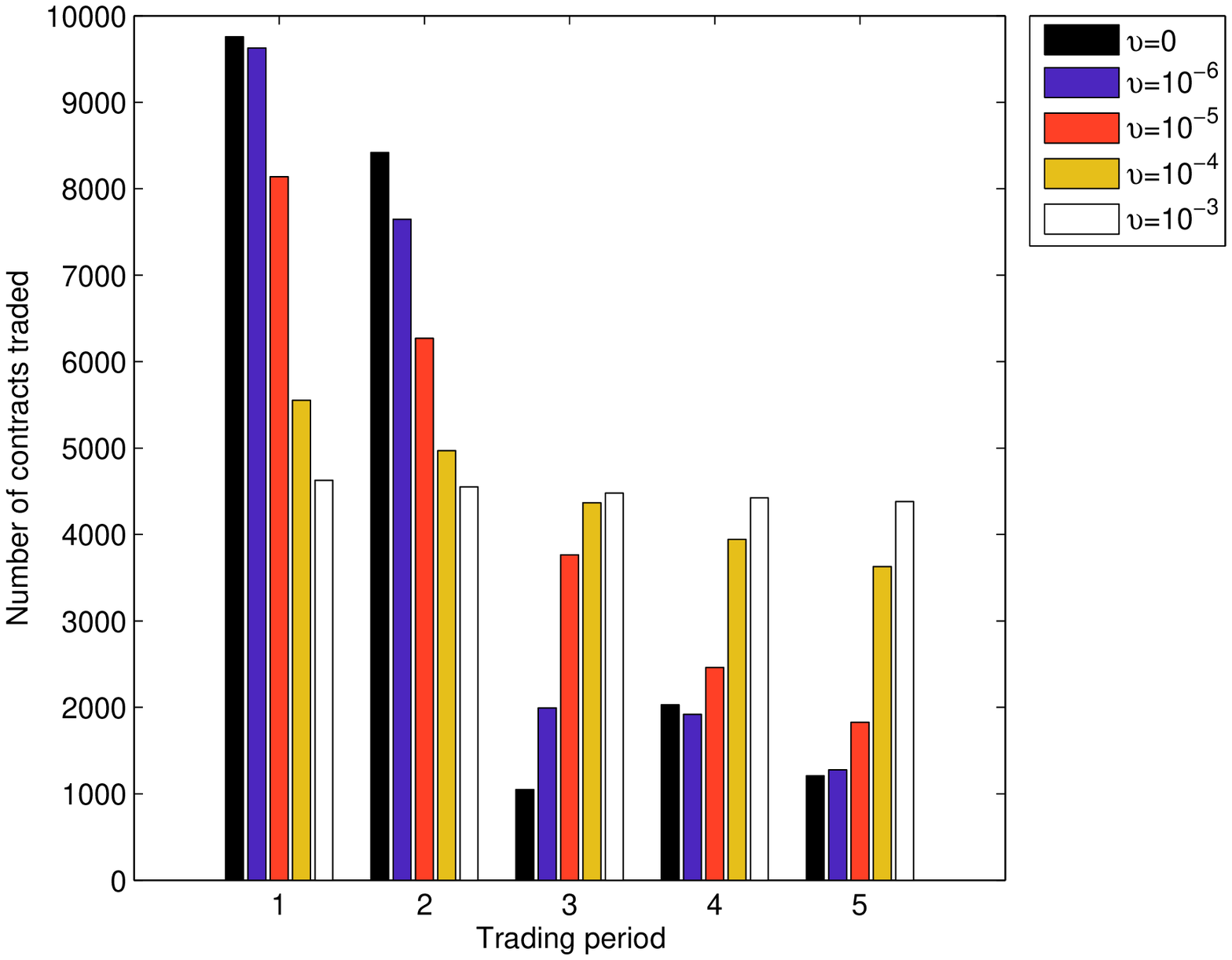}
\par\end{centering}

\caption{\label{fig:f4} $\lambda_{k}=10^{-6}$ for all $k\in P\cup C$, $\epsilon_{ij}=0$
and increasing $\upsilon_{ij}$ for all $i\in\left\{ 1,...,5\right\} $
and $j\in\left\{ 1\right\} $.}
\end{figure}

When $\upsilon$ is changed only for the first trading period (i.e.
the two month ahead contract), this significantly changes the term
structure of the price. As we can see in Figure \ref{fig:f5}, a large
increase in $\upsilon$ changes the term structure from normal backwardation
to contango (i.e. downward sloping). As expected, the players also
trade a much smaller number of contracts in the first trading period.

\begin{figure}
\begin{centering}
\includegraphics[bb=35bp 180bp 545bp 600bp,clip,scale=0.4]{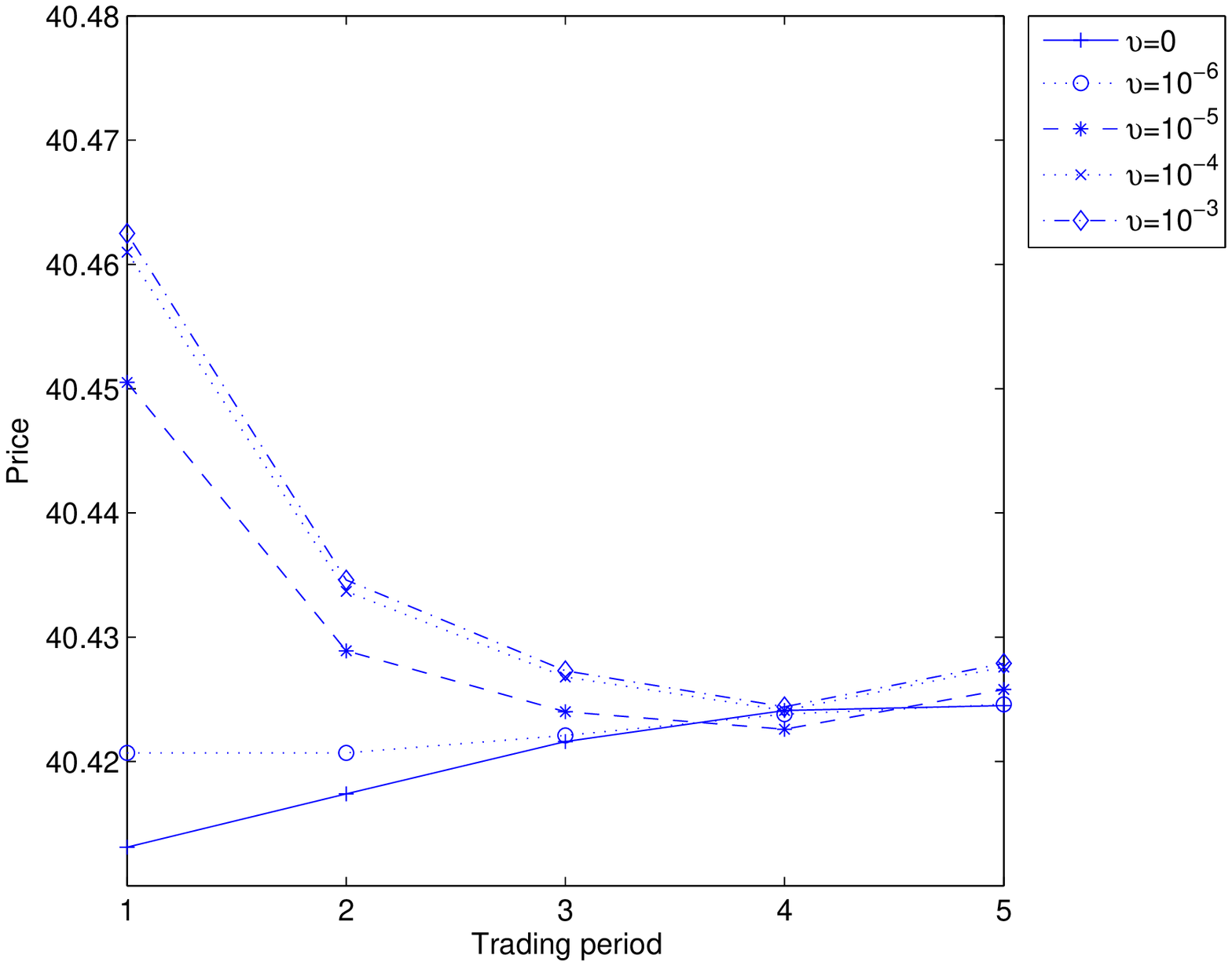}\includegraphics[bb=35bp 180bp 545bp 600bp,clip,scale=0.4]{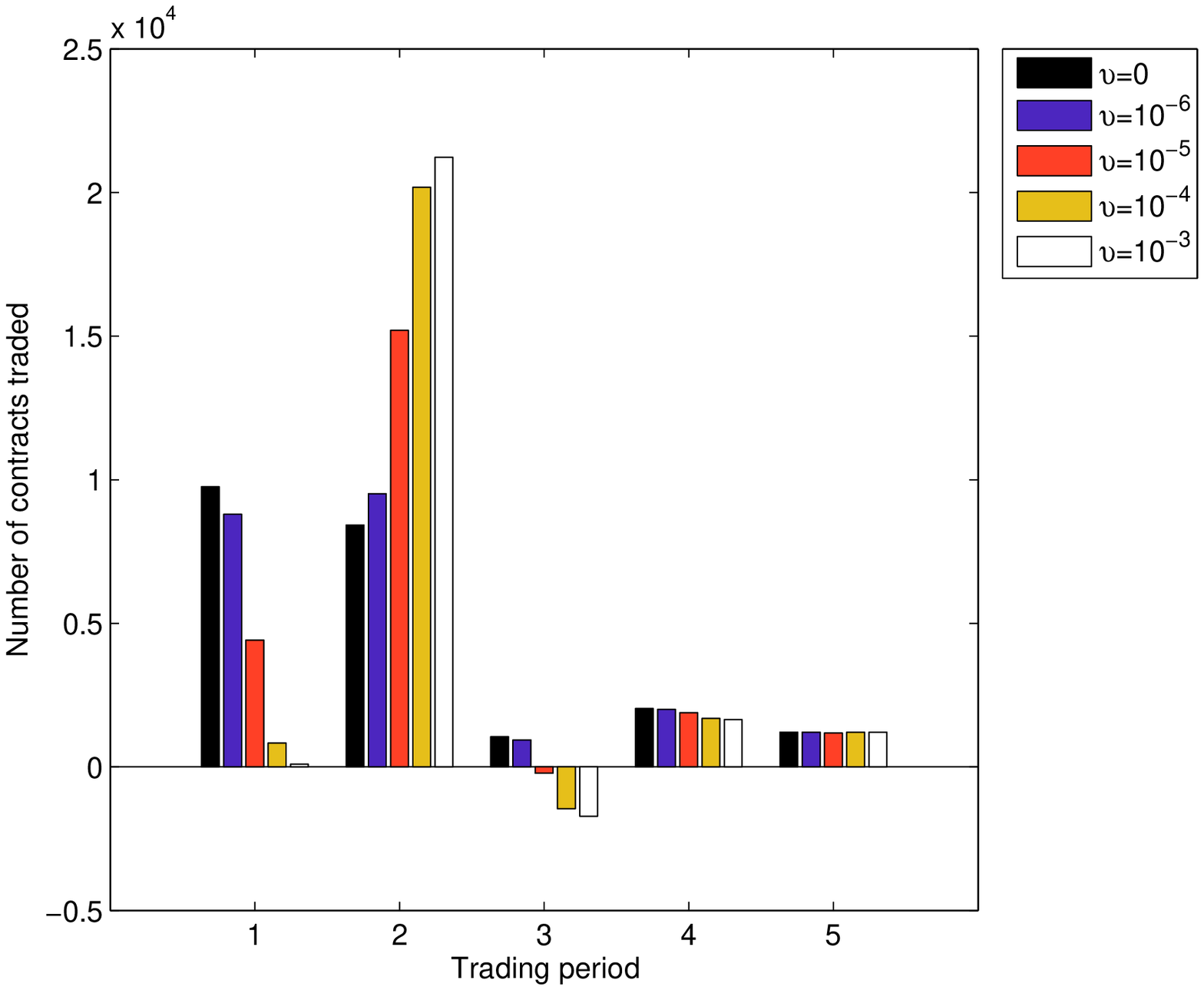}
\par\end{centering}

\caption{\label{fig:f5}$\epsilon_{ij}=0$ for all $i\in\left\{ 1,...,5\right\} $
and $j\in\left\{ 1\right\} $, $\lambda_{k}=10^{-6}$ for all $k\in P\cup C$,
and increasing $\upsilon_{11}$. For the other trading periods $i\in\left\{ 2,...,5\right\} $
and $j\in\left\{ 1\right\} $, $\upsilon_{ij}=0$.}
\end{figure}

An effect of a simultaneous change in the linear trading costs $\epsilon$
for all the contracts is depicted in Figure \ref{fig:f6}. We can
see that the linear trading costs have no impact on the trading strategy
and on the shape of the term structure. However, the price of all
contracts is increased to cover the losses of the producers caused
by the increased linear cost of trading. The consumers cover the losses
caused by the increased linear cost of trading through increased electricity
prices they charge to the end users. The simultaneous change in the
linear trading costs $\epsilon$ for all the contracts does not impact
the optimal trading strategy of the players.

\begin{figure}
\begin{centering}
\includegraphics[bb=35bp 180bp 545bp 600bp,clip,scale=0.4]{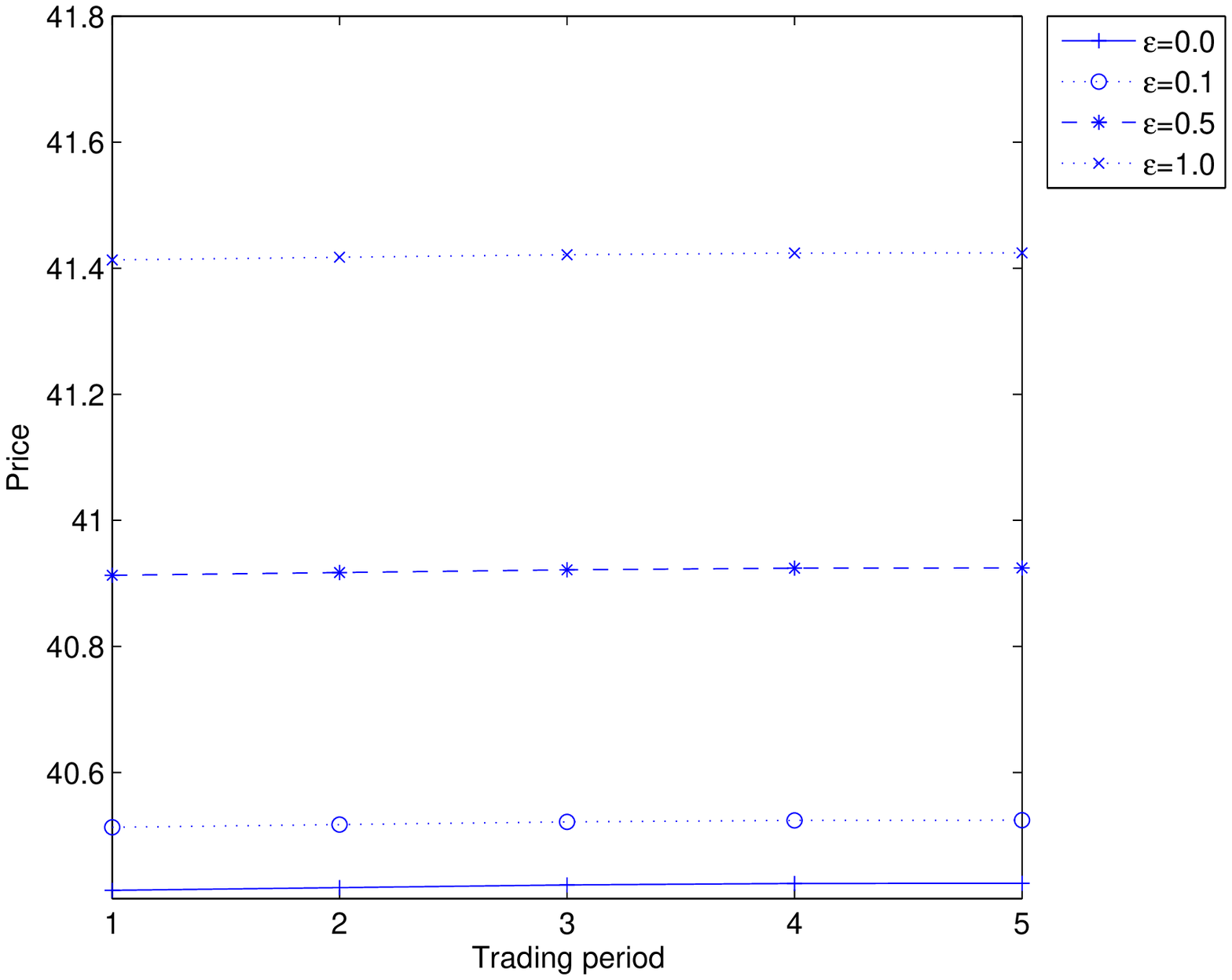}\includegraphics[bb=35bp 180bp 545bp 600bp,clip,scale=0.4]{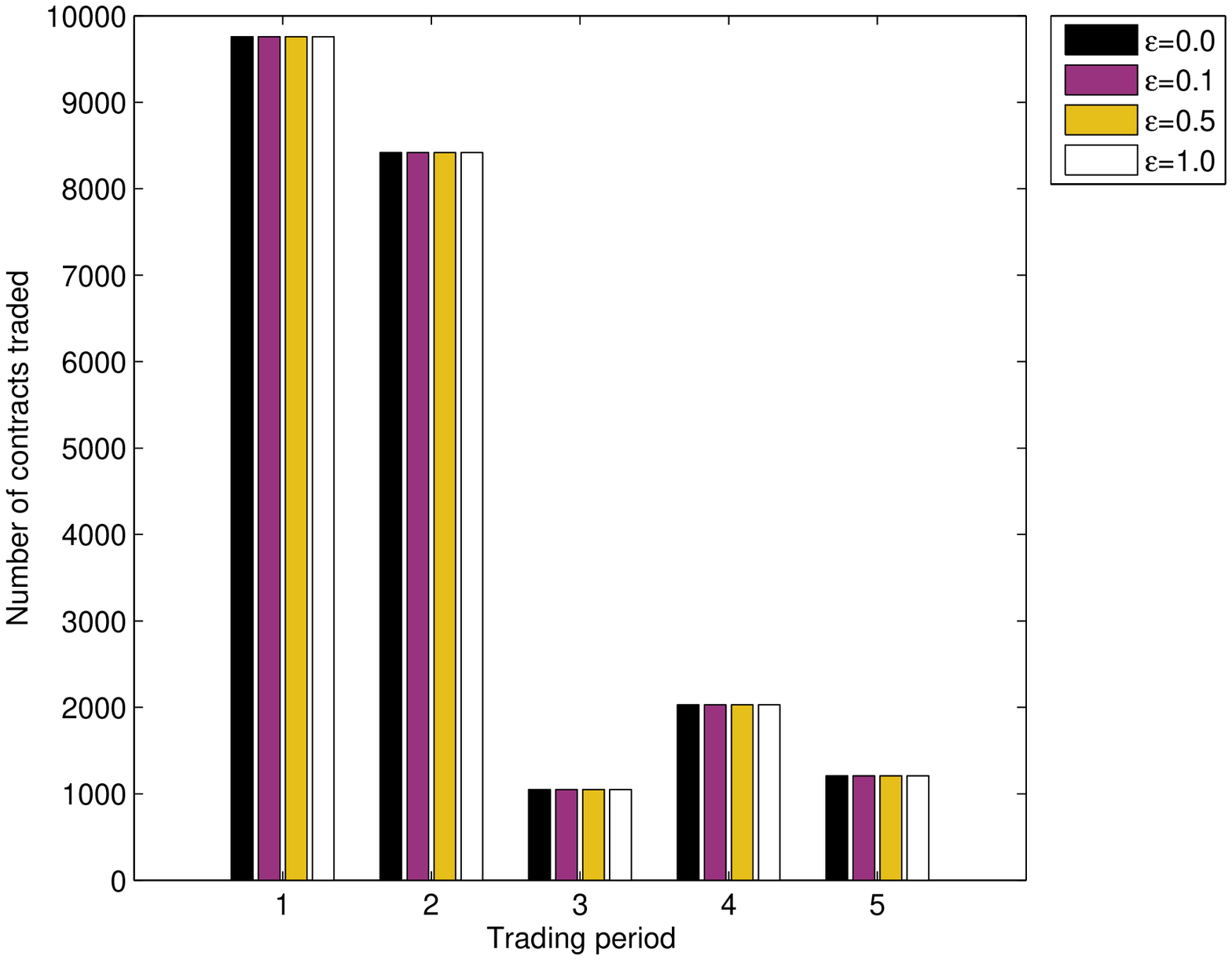}
\par\end{centering}

\caption{\label{fig:f6} $\lambda_{k}=10^{-6}$ for all $k\in P\cup C$, $\upsilon_{ij}=0$
and increasing $\epsilon_{ij}$ for all $i\in\left\{ 1,...,5\right\} $
and $j\in\left\{ 1\right\} $.}
\end{figure}

When $\epsilon$ is changed only for the first trading period (i.e.
the two month ahead contract), this, in contrast to the changes in
$\upsilon$, significantly affects only the price in the first trading
period. The price of the two month ahead contract is increased to
cover the losses of the producers caused by the increased linear cost
of trading. If $\epsilon$ is increased significantly, this makes
the two month ahead contract so unattractive that no contracts are
traded (see the right part of Figure \ref{fig:f7}).

\begin{figure}
\begin{centering}
\includegraphics[bb=35bp 180bp 545bp 600bp,clip,scale=0.4]{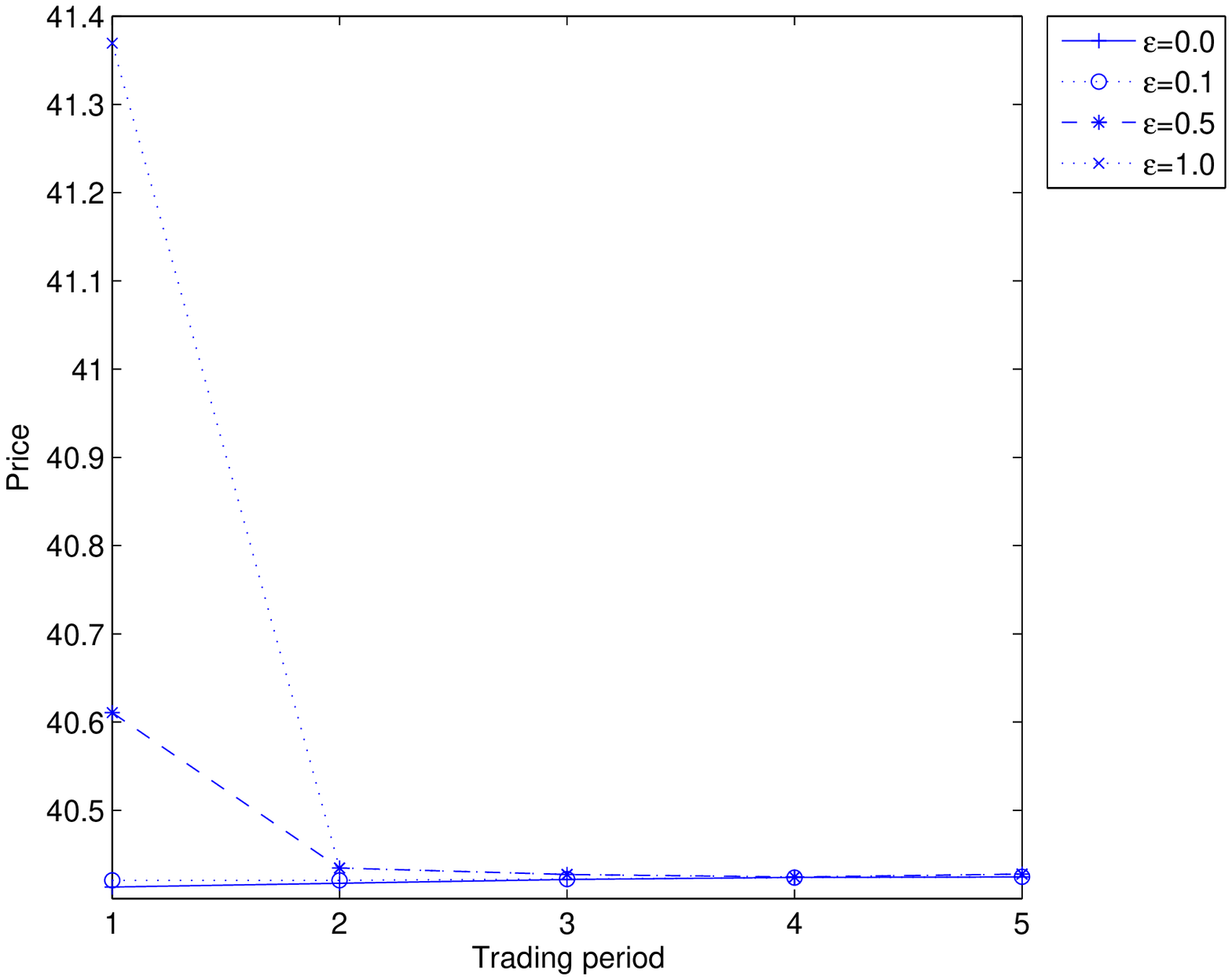}\includegraphics[bb=35bp 180bp 545bp 600bp,clip,scale=0.4]{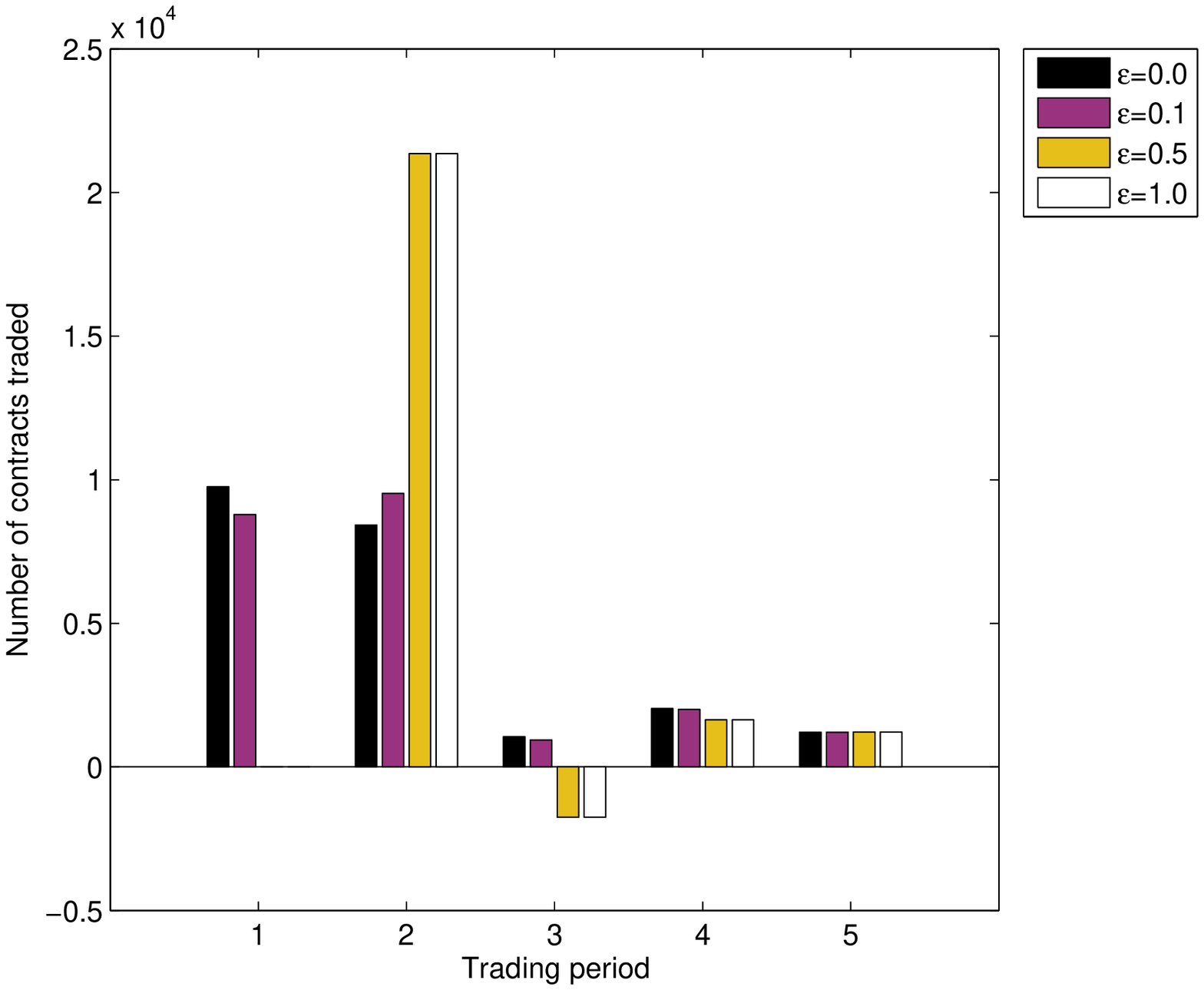}
\par\end{centering}

\caption{\label{fig:f7}$\upsilon_{ij}=0$ for all $i\in\left\{ 1,...,5\right\} $
and $j\in\left\{ 1\right\} $, $\lambda_{k}=10^{-6}$ for all $k\in P\cup C$,
and increasing $\epsilon_{11}$. For the other trading periods $i\in\left\{ 2,...,5\right\} $
and $j\in\left\{ 1\right\} $, $\epsilon_{ij}=0$.}
\end{figure}

In the reminder of this section, we investigate the effect of the
term structure of fuel and emission prices on the term structure of
electricity prices. Initially, we set
\begin{equation}
\begin{array}{ccccccccc}
\mathbb{E}^{\mathbb{P}}\left[G_{1}\left(t_{1},T_{1}\right)\right] & = & \mathbb{E}^{\mathbb{P}}\left[G_{1}\left(t_{2},T_{1}\right)\right] & = & \ldots & = & \mathbb{E}^{\mathbb{P}}\left[G_{1}\left(t_{5},T_{1}\right)\right] & = & 69.30\:\left[\text{GBp/therm}\right],\\
\\
\mathbb{E}^{\mathbb{P}}\left[G_{2}\left(t_{1},T_{1}\right)\right] & = & \mathbb{E}^{\mathbb{P}}\left[G_{2}\left(t_{2},T_{1}\right)\right] & = & \ldots & = & \mathbb{E}^{\mathbb{P}}\left[G_{2}\left(t_{5},T_{1}\right)\right] & = & 57.87\:\left[\text{GBP/tonne}\right],\\
\\
\mathbb{E}^{\mathbb{P}}\left[G_{em}\left(t_{1},T_{1}\right)\right] & = & \mathbb{E}^{\mathbb{P}}\left[G_{em}\left(t_{2},T_{1}\right)\right] & = & \ldots & = & \mathbb{E}^{\mathbb{P}}\left[G_{em}\left(t_{5},T_{1}\right)\right] & = & 3.883\:\left[\text{GBP/tonne}\right],
\end{array}
\end{equation}
where $l=1$ and $l=2$ denote gas and coal prices, respectively.
Electricity price is quoted in GBP/MWh.

We conducted two types of experiments: In the first type, we performed
parallel shifts of the term structure of fuel/emission prices. In
the second type, we changed the shape of the term structure of fuel/emission
prices. Since results for all fuels as well as for emissions are very
similar, we present them for gas only.

Figure \ref{fig:f7-1} depicts the effect of parallel shifts in the
term structure of gas prices. Expectedly, an increase/decrease in
gas prices causes an increase/decrease in electricity prices.

\begin{figure}
\begin{centering}
\includegraphics[bb=35bp 180bp 545bp 600bp,clip,scale=0.4]{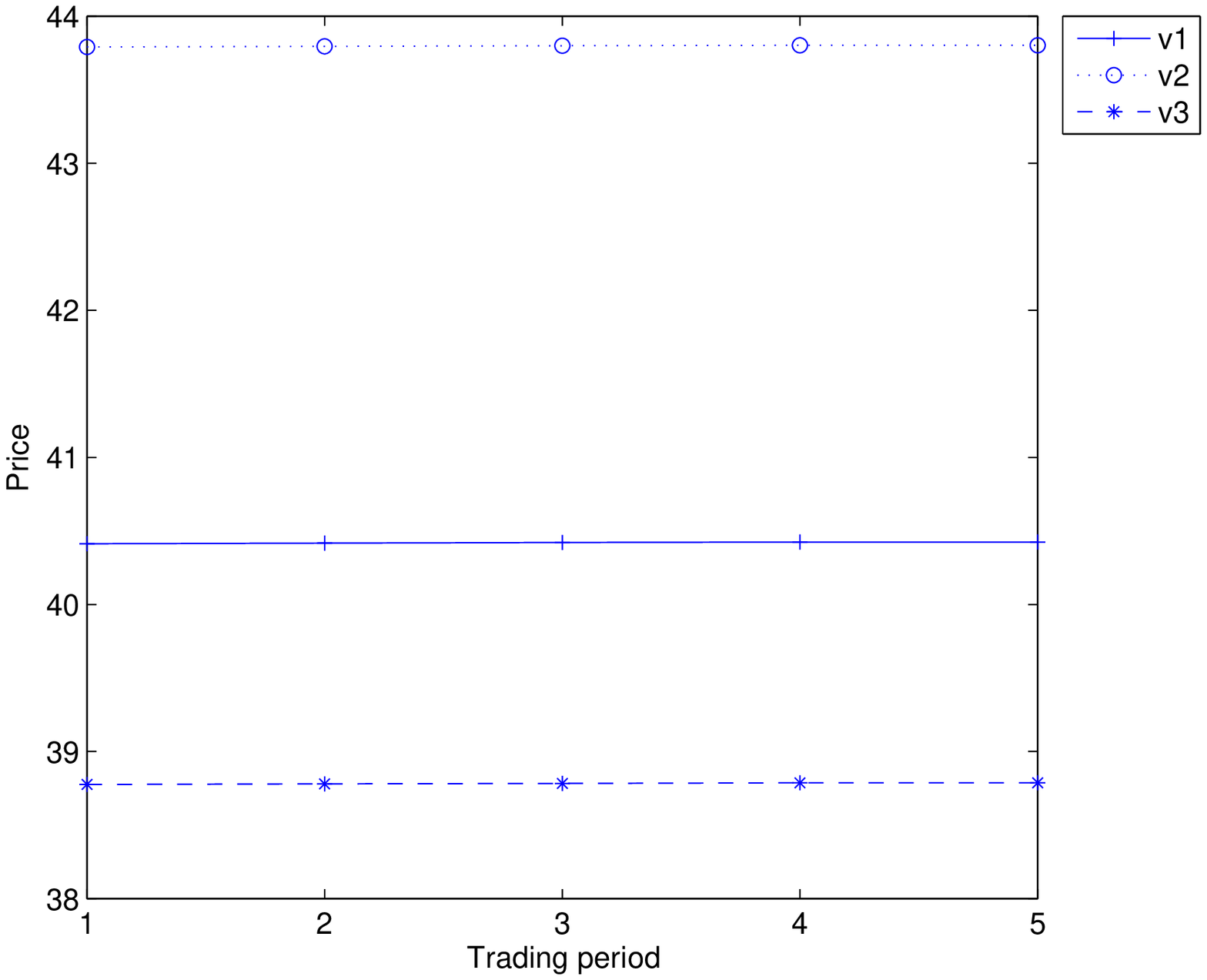}\includegraphics[bb=35bp 180bp 545bp 600bp,clip,scale=0.4]{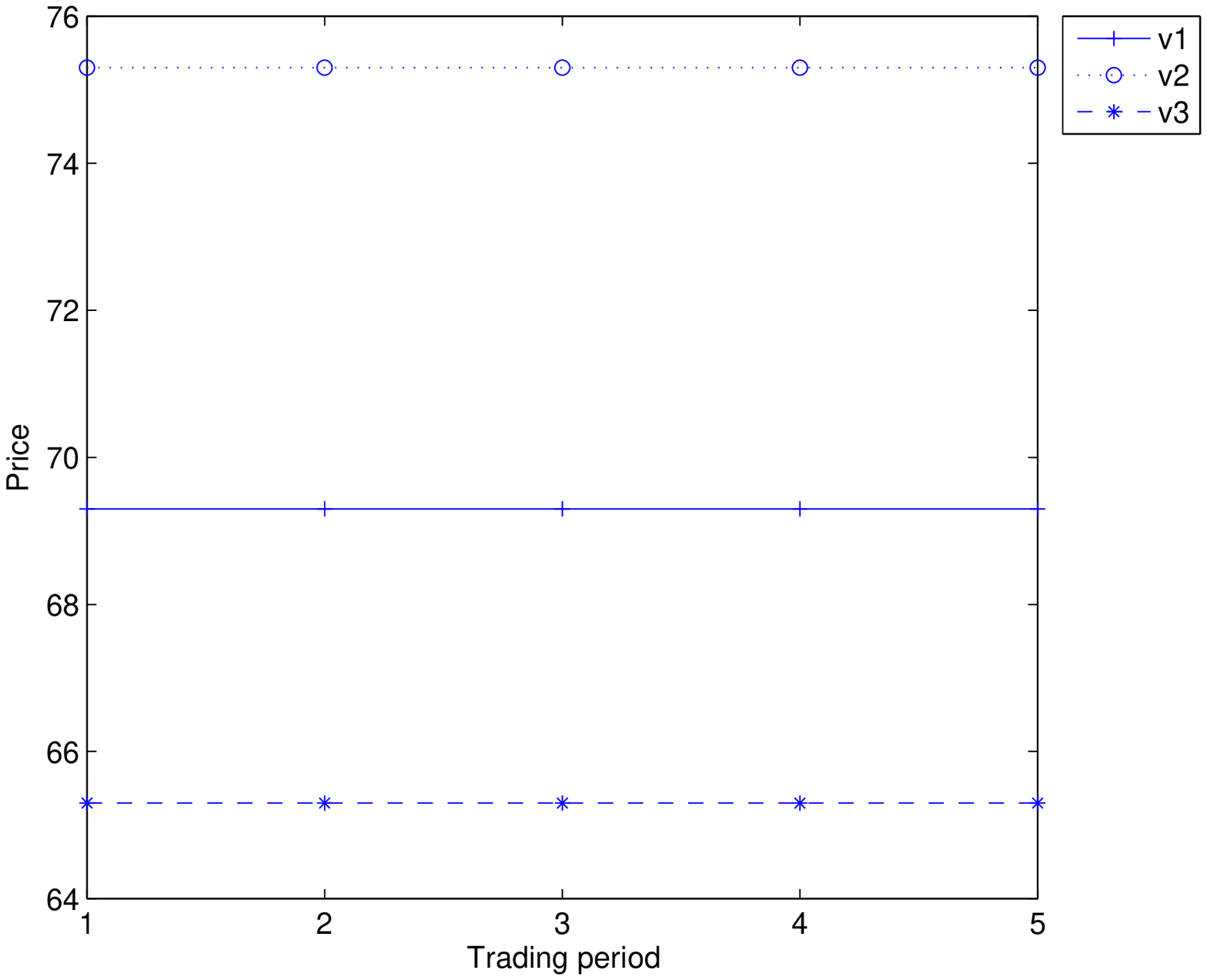}
\par\end{centering}

\caption{\label{fig:f7-1}Effect of parallel shifts in the term structure of
gas prices (right) on the term structure of electricity prices (left). }
\end{figure}

More interesting is the second experiment, where we calculated the
term structure of electricity prices for three different shapes of
the term structure of gas prices. In Figure \ref{fig:f7-1-1}, we
examined a constant term structure of gas prices, normal backwardation
and contango. We can see that the term structure of gas/emission prices
has a large impact on the term structure of electricity prices.

\begin{figure}
\begin{centering}
\includegraphics[bb=35bp 180bp 545bp 600bp,clip,scale=0.4]{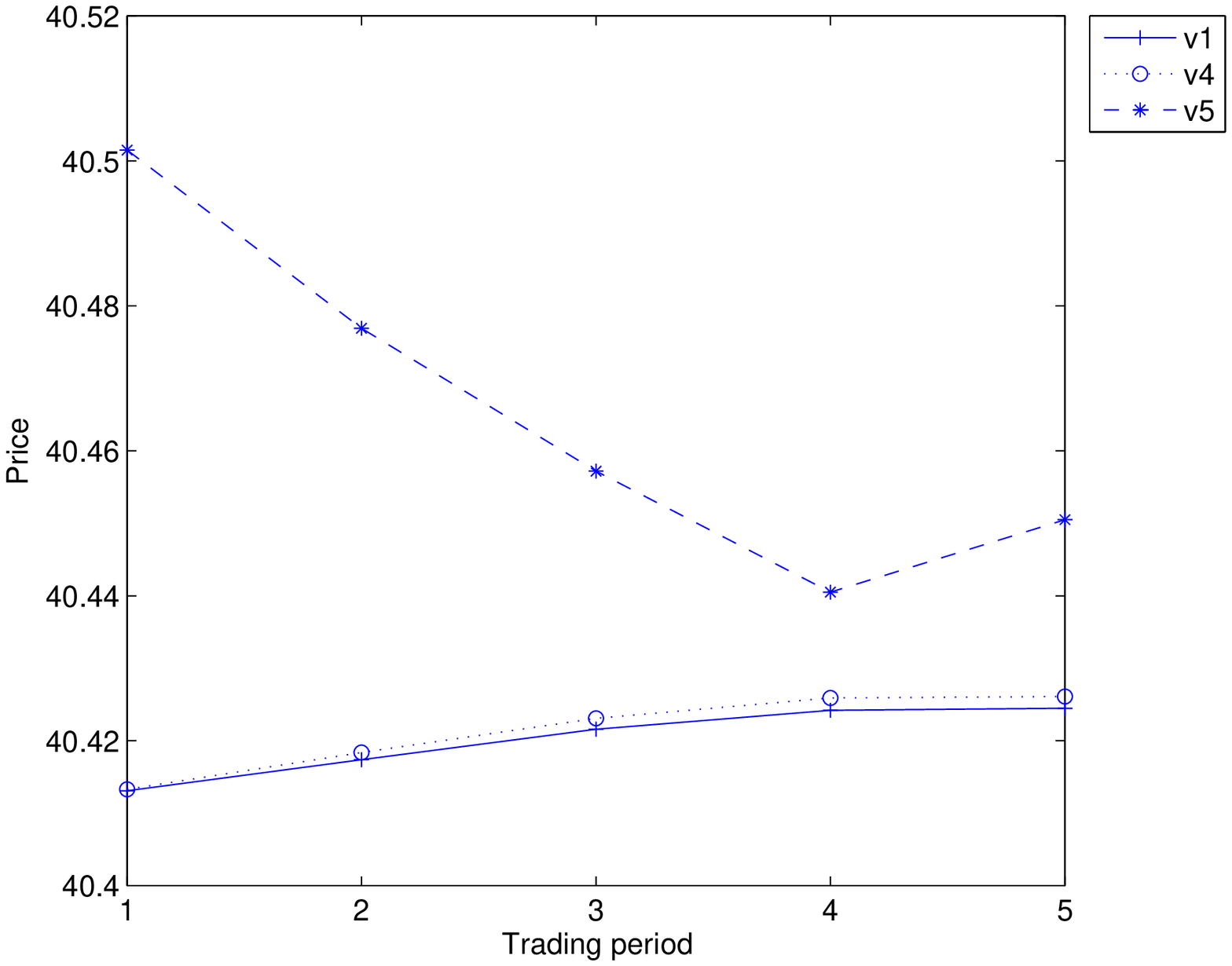}\includegraphics[bb=35bp 180bp 545bp 600bp,clip,scale=0.4]{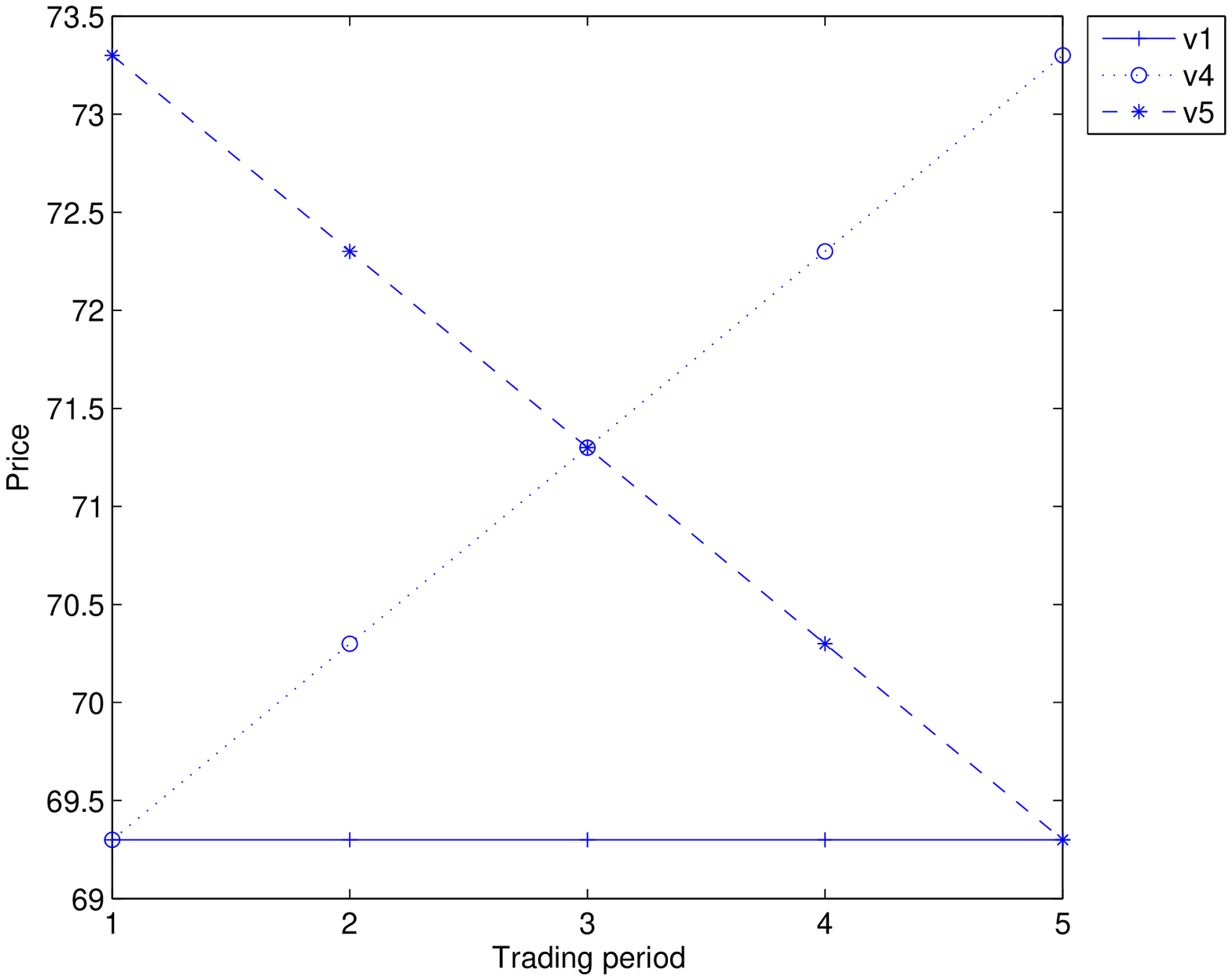}
\par\end{centering}

\caption{\label{fig:f7-1-1}Dependence of the term structure of electricity
prices (left) on the shape of the term structure of gas prices (right).}
\end{figure}

\subsection{The UK power system}

In this subsection we apply our model to the entire system of the
UK power plants. We focus on the coal, gas, and oil power plants,
because these power plants adapt their production to cover the changes
in demand and are thus responsible for setting the price. Nuclear
power plants do not have to be modeled explicitly because their ramp
up and ramp down constraints are so tight that their production is
almost constant over time. They usually deviate from the maximum production
only for maintenance reasons. Renewable sources and interconnectors
are not modeled explicitly, because they require a different treatment
not covered in this paper. In this subsection, we define demand $D\left(T_{j}\right)$
for all $j\in J$ as
\begin{equation}
D\left(T_{j}\right):=D_{act}\left(T_{j}\right)-P_{renw}\left(T_{j}\right)-P_{inter}\left(T_{j}\right)
\end{equation}
where $D_{act}\left(T_{j}\right)$ denotes the actual demand in the
UK power system, $P_{renw}\left(T_{j}\right)$ denotes the production
from all renewable sources including wind, solar, biomass, hydro and
pumped storage, and $P_{inter}\left(T_{j}\right)$ denotes the inflow
of power in to the UK power system through interconnectors. To make
this model useful in practice one has to model each of these terms,
but this exceeds the scope of this paper.

We calibrate the model and for each power plant $r\in R^{p,l}$, $l\in L$
and $p\in P$, we estimate the efficiency $c^{p,l,r}$, the carbon
emission intensity factor $g^{p,l,r}$, the maximum capacity $\overline{W}_{max}^{p,l,r}$,
the ramp up rate $\triangle\overline{W}_{max}^{p,l,r}$ and the ramp
down rate $\triangle\overline{W}_{min}^{p,l,r}$ using the historical
production as described in Subsection \ref{sub:Calibration}. The
covariance matrix was also estimated from the historical data using
the shrinkage approach described in \cite{ledoit2003improved}. 

Our goal is to calculate the electricity spot price with the information
available on 11/2/2013. We are interested in a delivery period from
4/4/2013 00:00:00 to 8/4/2013 00:00:00. We assume that there are two
types of power contract available. The first is a month ahead contract
traded on 15/3/2013 17:00:00 and covers the delivery over all four
days. The second type is a spot contract that requires an immediate
delivery and is traded for each half hour separately. We use future
prices of coal, gas, and oil as available on 11/2/2013. Since the
historical demand forecast is not available, we used the realized
demand instead, which is a standard practice in the literature. To
use this model in practice, one could use a demand forecast available
at the Elexon web page or develop a new approach. Since we do not
have the information about the ownership of the power plants, we assumed
that there is only one producer who owns all power plants connected
to the UK grid and only one consumer that is responsible for satisfying
the demand of the end users. In reality, market participants have
more information about the ownership that can be incorporated into
the model.

The numerical results in Figures \ref{fig:f8} - \ref{fig:f15} are
all calculated using $\lambda_{k}=10^{-5}$ for all $k\in P\cup C$,
and $\epsilon_{ij}=0.1$ and $\upsilon_{ij}=10^{-4}$ for all contracts.
The figure on the left hand side depicts the calculated energy mix
between coal and gas power plants, while the figure on the right hand
side depicts the actually observed energy mix. Both figures contain
also the spot price calculated by our model and the actually observed
spot price.

Figure \ref{fig:f8} shows that our model predicts the energy mix
very closely. Moreover, the daily pattern of the electricity price
predicted by our model is similar to the actually observed one. 

\begin{figure}
\begin{centering}
\includegraphics[bb=35bp 180bp 545bp 600bp,clip,scale=0.43]{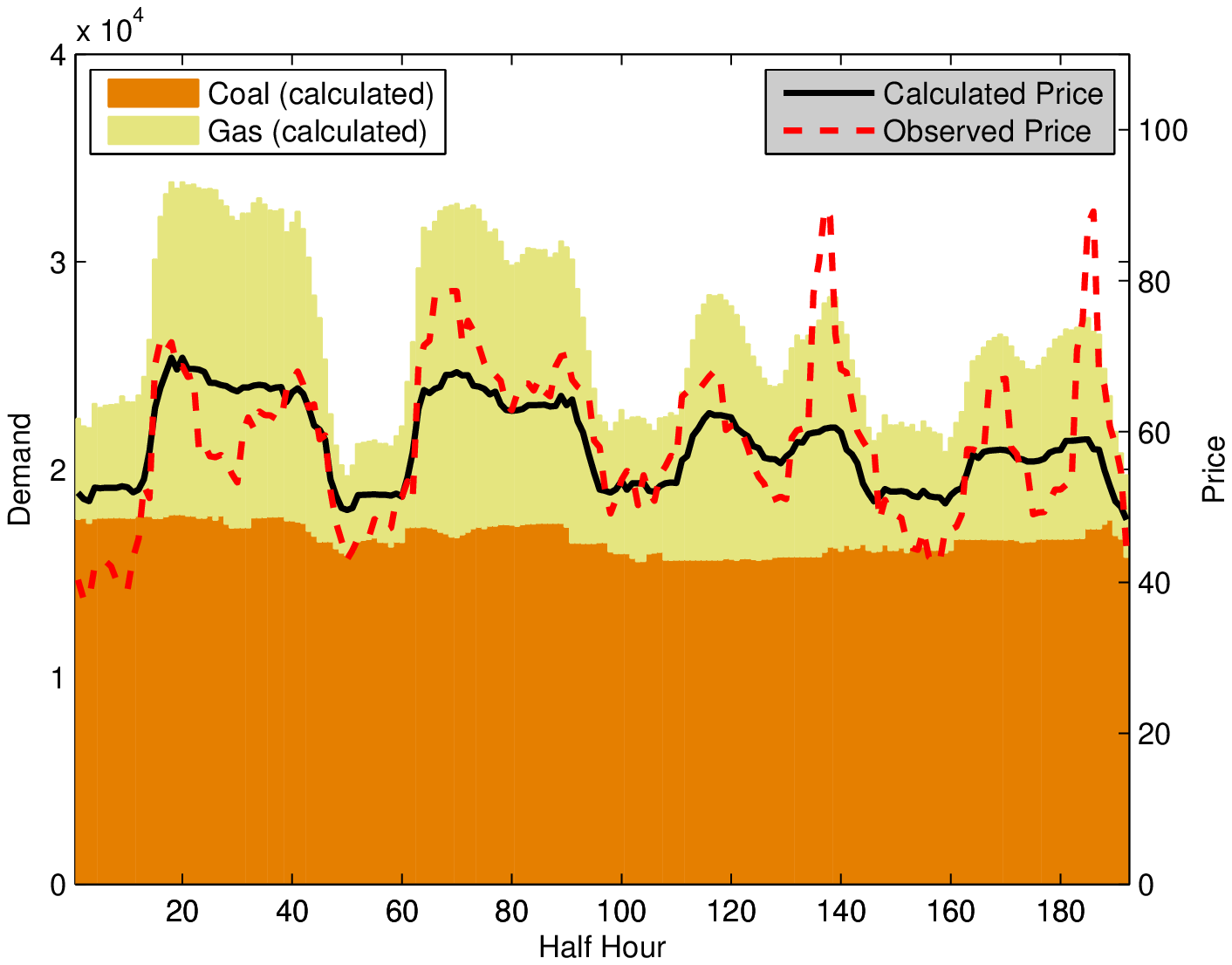}\includegraphics[bb=35bp 180bp 545bp 600bp,clip,scale=0.43]{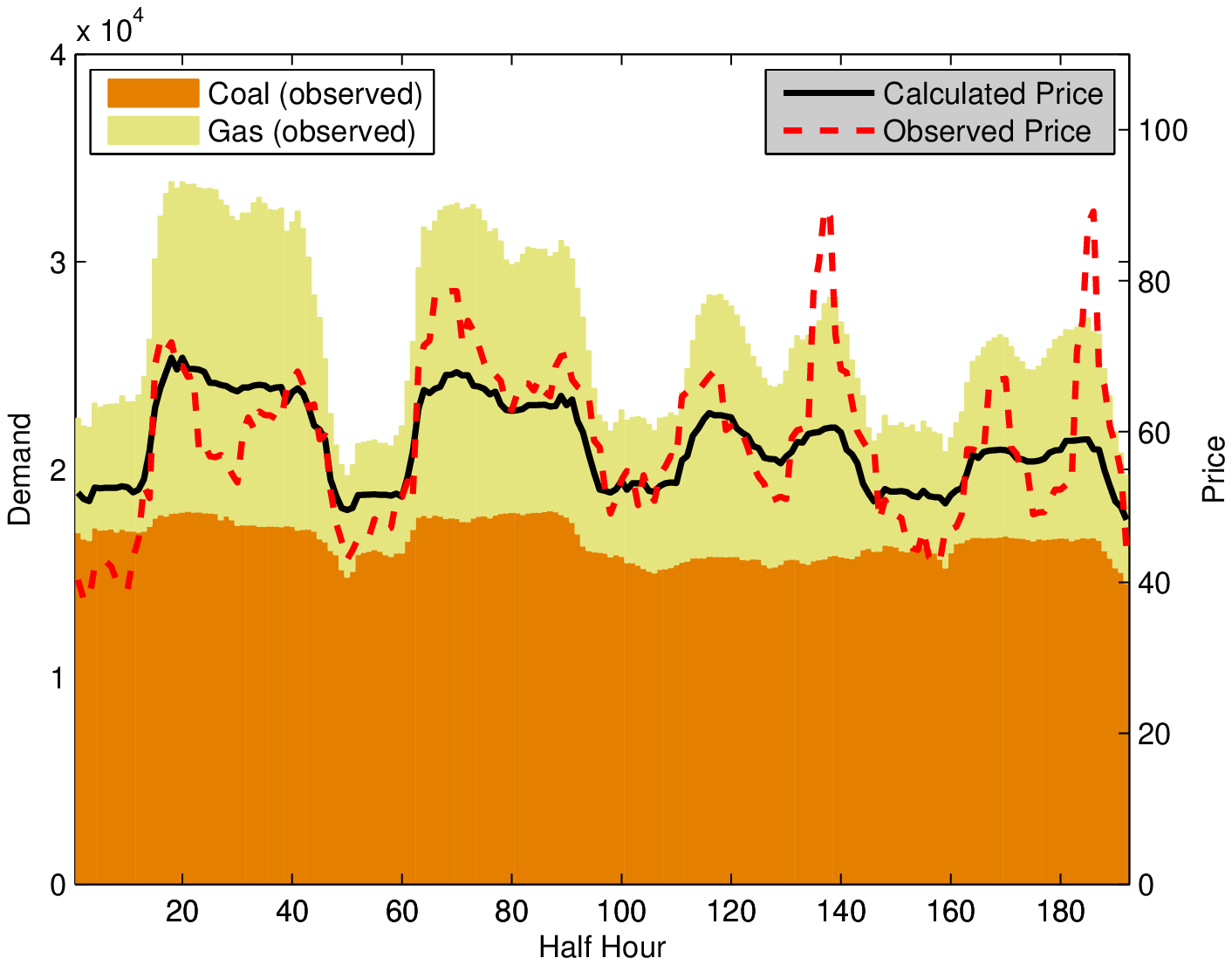}
\par\end{centering}

\caption{\label{fig:f8}$\lambda_{k}=10^{-7}$ for all $k\in P\cup C$, $\epsilon=0.1$,
and $\upsilon=10^{-4}$.}
\end{figure}

Figure \ref{fig:f11} shows how the results change when we tighten
the ramp up and ramp down constraints by 20\%. We can see that the
price is more serrated and slightly higher, because more expensive
power plants must be turned on to cover the changes in demand. 

\begin{figure}
\begin{centering}
\includegraphics[bb=35bp 180bp 545bp 600bp,clip,scale=0.43]{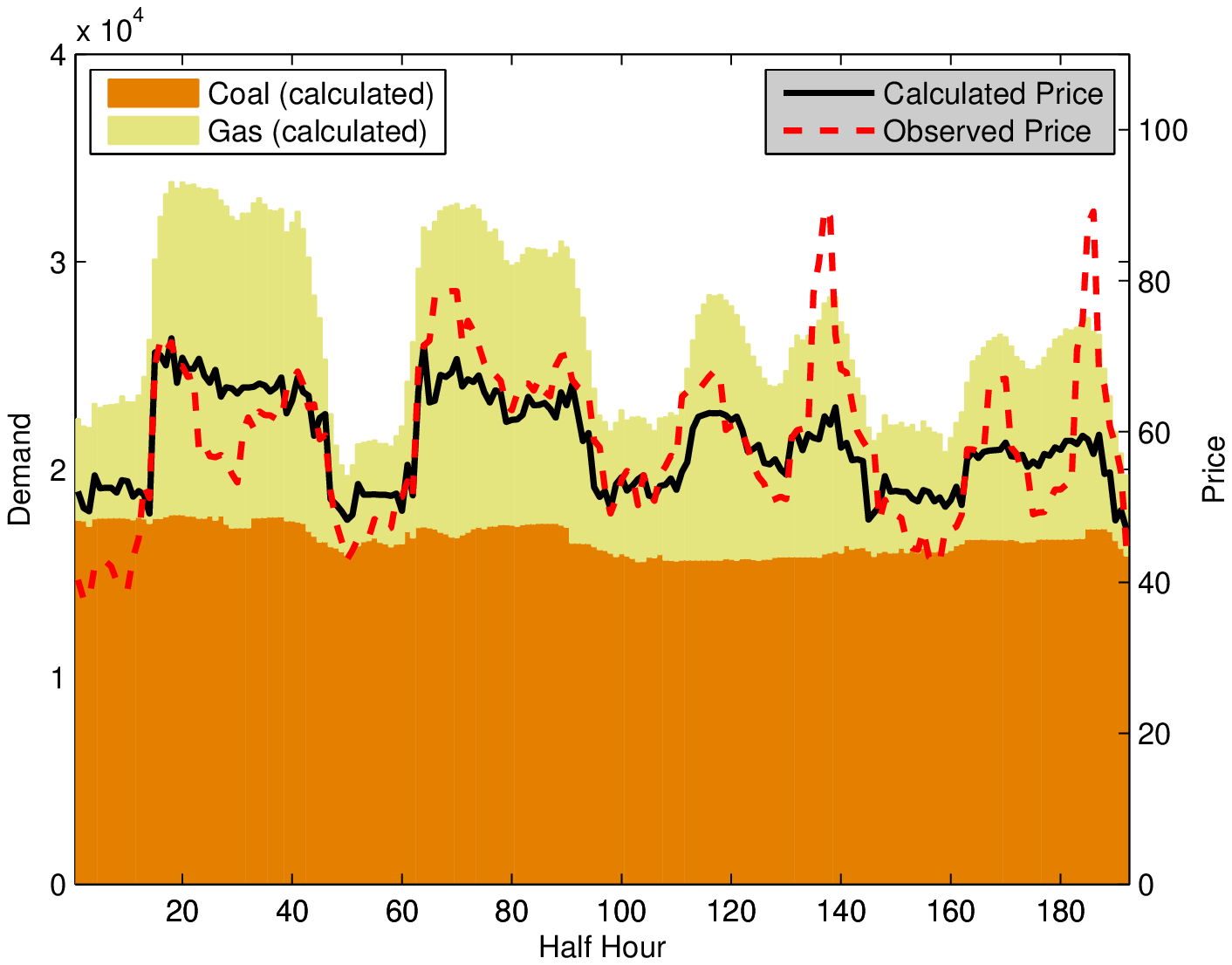}\includegraphics[bb=35bp 180bp 545bp 600bp,clip,scale=0.43]{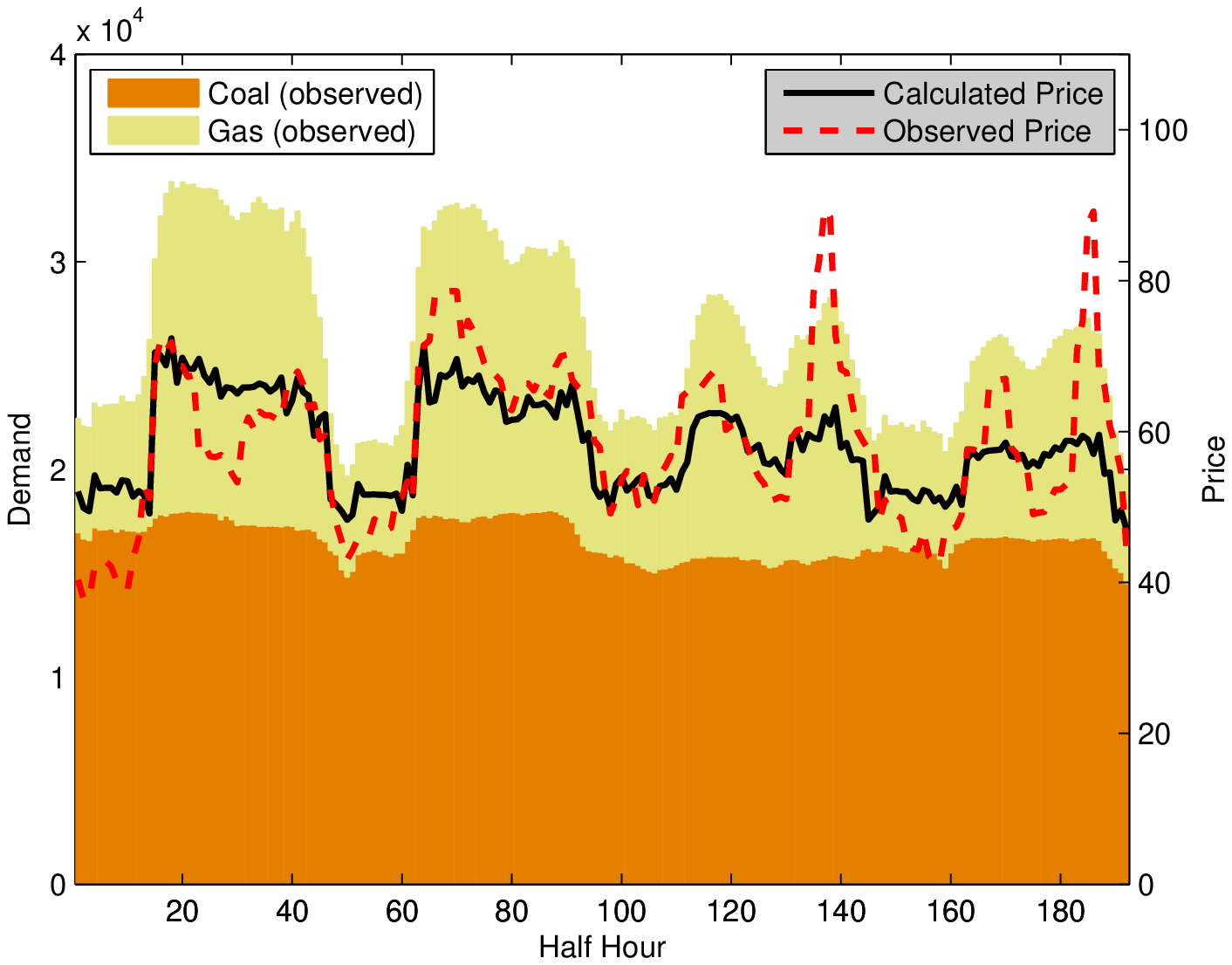}
\par\end{centering}

\caption{\label{fig:f11}Ramp up and ramp down constraints are tightened by
20\%.}
\end{figure}

Similarly, Figure \ref{fig:f12} shows how the results change when
we tighten the ramp up and ramp down constraints by 50\%. 

\begin{figure}
\begin{centering}
\includegraphics[bb=35bp 180bp 545bp 600bp,clip,scale=0.43]{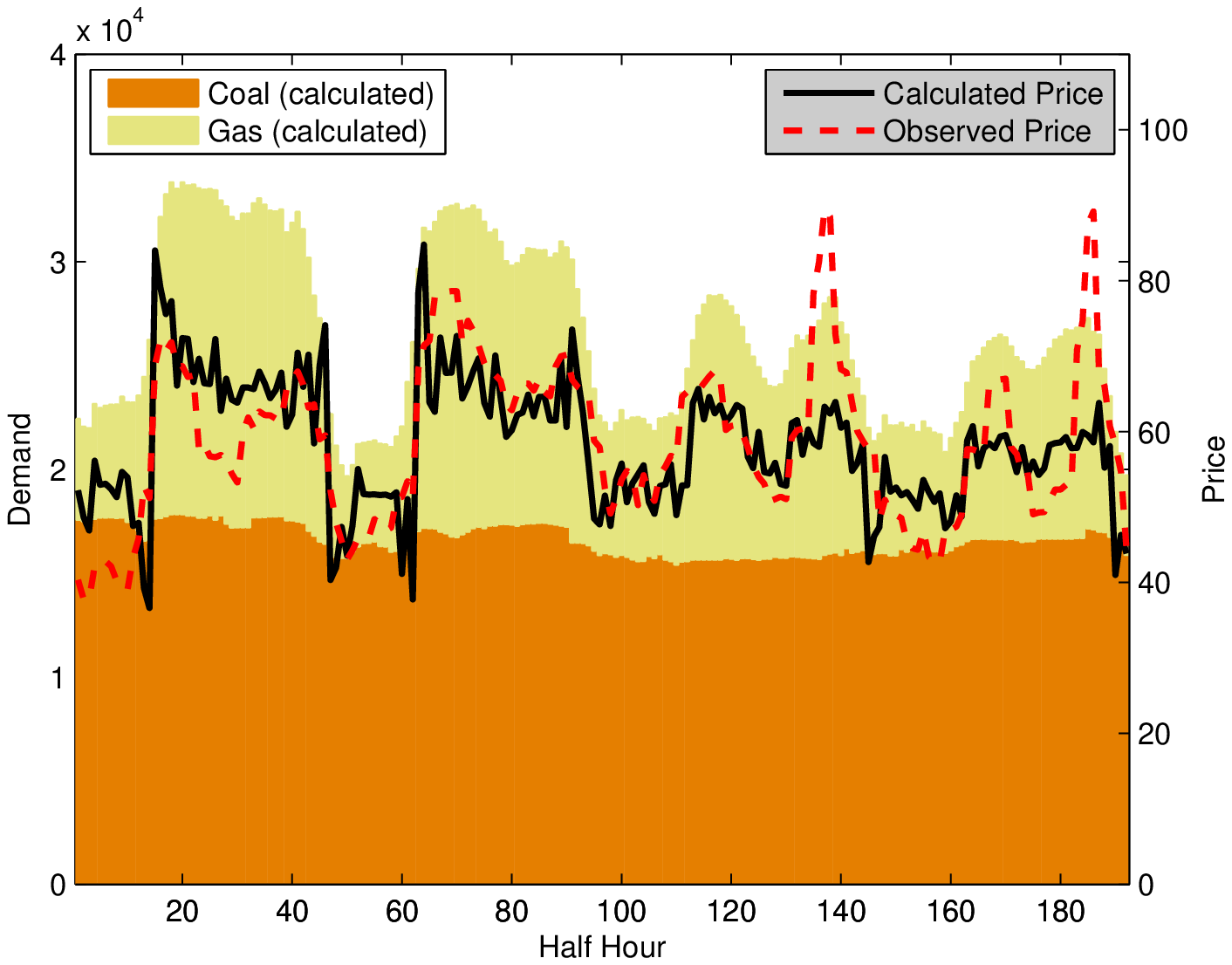}\includegraphics[bb=35bp 180bp 545bp 600bp,clip,scale=0.43]{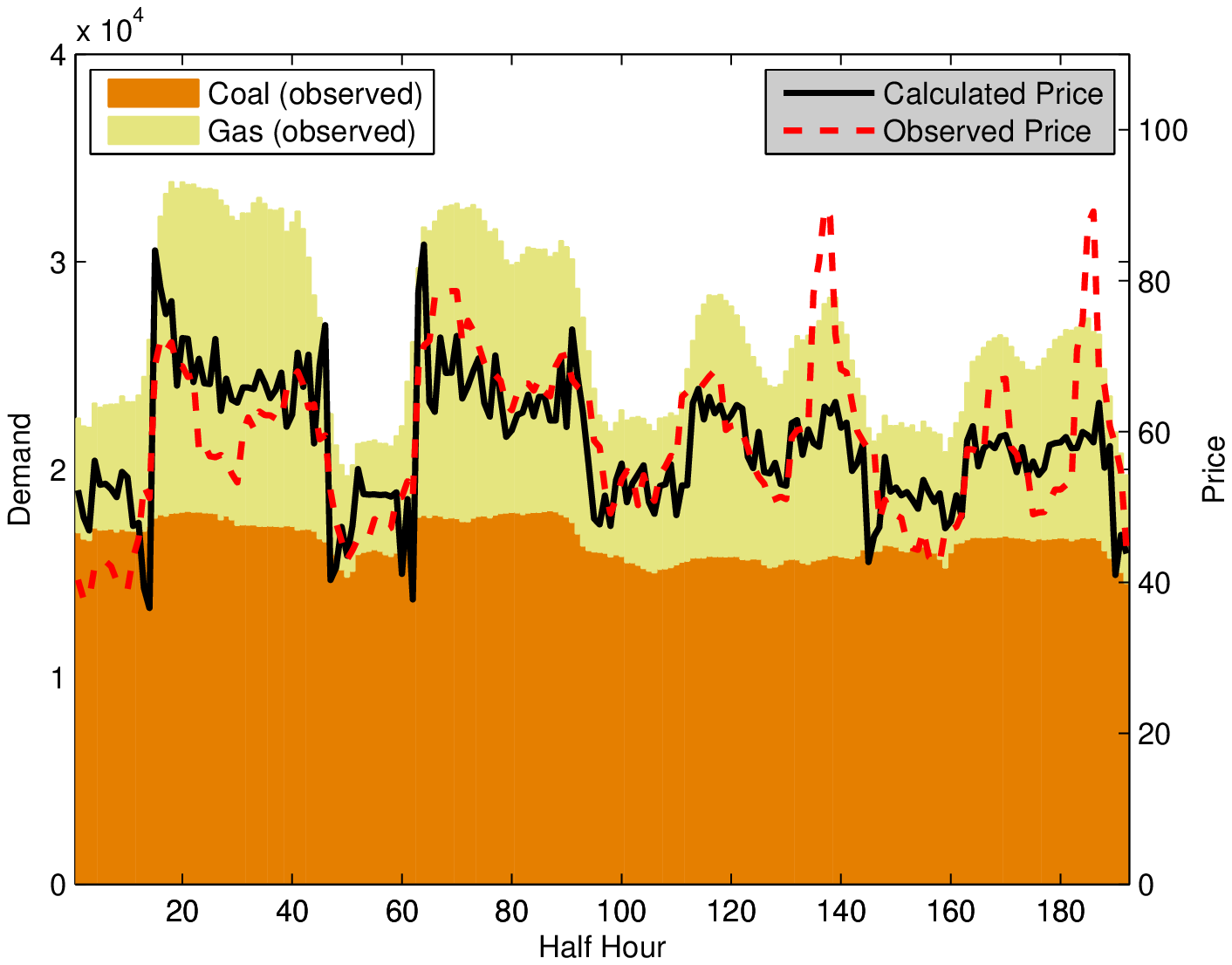}
\par\end{centering}

\caption{\label{fig:f12}Ramp up and ramp down constraints are tightened by
50\%.}
\end{figure}

Figure \ref{fig:f8} reveals a few problems of our model. Firstly,
we can see that the daily variation in the calculated price is smaller
than in the observed one. Secondly, the two spikes in the observed
price are not captured in our model. Our hypothesis is that the first
problem occurs, because our model does not incorporate the start-up
costs of the power plants correctly. Inclusion of the start-up costs
exceeds the scope of this paper and will be addressed separately.
Some preliminary results are shown in Figure \ref{fig:f15}. Our hypothesis
for the explanation of the second problem is the error in the demand
forecast. Spikes in the spot price occur when there is an unexpected
change in demand. Since only a few (usually rather inefficient Open
Cycle Gas Turbine) power plants are flexible enough to cover the demand,
they require a high electricity price to be turned on. Thus, the spikes
cannot be forecast two months before the delivery. An investigation
of the predictive power of our model to forecast spikes closer to
delivery is left for future work.

\begin{figure}
\begin{centering}
\includegraphics[bb=35bp 180bp 545bp 600bp,clip,scale=0.43]{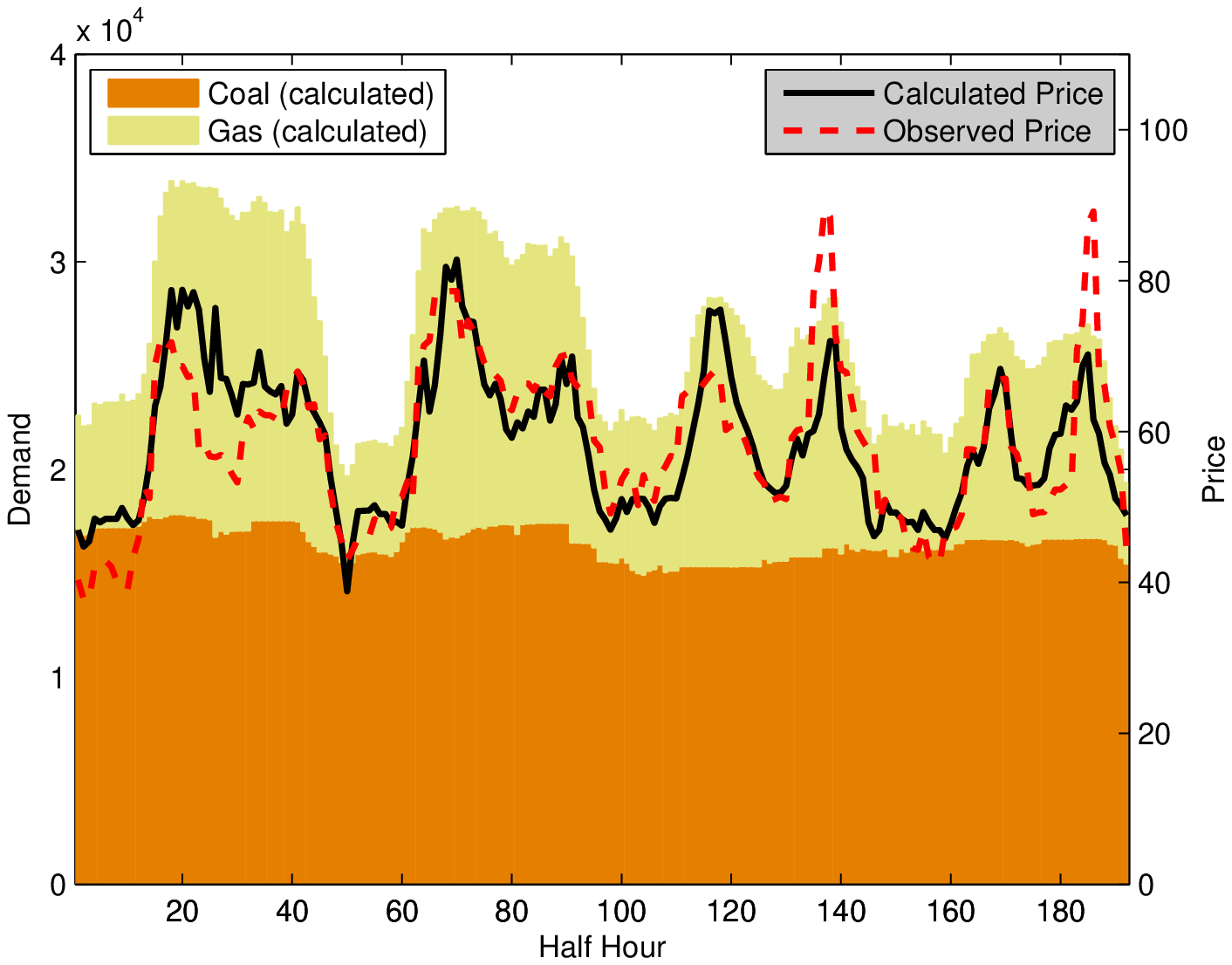}\includegraphics[bb=35bp 180bp 545bp 600bp,clip,scale=0.43]{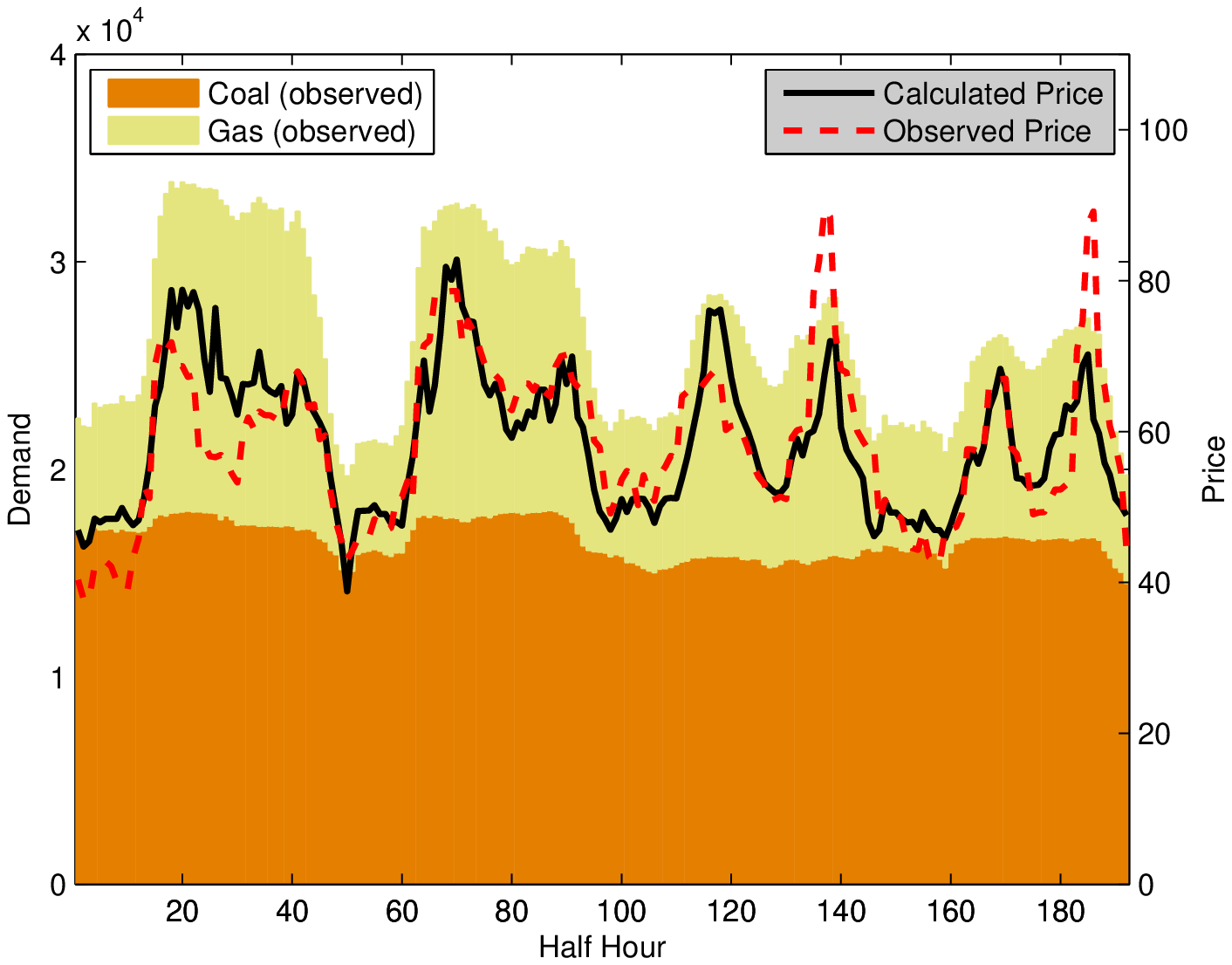}
\par\end{centering}

\caption{\label{fig:f15}Preliminary result after including the startup costs.}
\end{figure}

\section{Conclusions\label{sec:Conclusions}}

In this paper we proposed a tractable quadratic programming reformulation
for calculating the equilibrium term structure of electricity prices
in a market with multiple consumers and producers. Reformulation can
be used for solving general equilibrium optimization problems that
include inequality constraints and can be seen as an extension of
the shadow price approach. 

We have extended the term structure electricity price model proposed
in \cite{troha2014theexistence} to a more realistic setting by taking
into account transaction costs and liquidity considerations as well
as realistic electricity contracts. 

The section on numerical results first shows how to calibrate various
parameters of the model. We investigated how different parameters
affect the equilibrium electricity price. Interestingly, we found
that consumers have very little power to influence the electricity
price. They have to satisfy the demand of the end users regardless
of the price level. To maintain their profitability, they propagate
all increased costs to the end users. Producers, on the other hand,
have much more power to affect the electricity prices. They have a
large impact on the price level and a slightly smaller impact on the
term structure itself. The market micro-structure and liquidity of
the contracts can significantly affect the term structure of the electricity
price. In an extreme case, changes in the liquidity even altered the
term structure from normal backwardation to contango. The term structure
of electricity prices is also considerably affected by the term structure
of fuels and emissions.

We investigated the predictive power of our model when applied to
the realistic system of UK power plants. The results show that our
model predicts the prices quite well and that the predicted price
exhibits the main features of the electricity price. Numerical examples
show the effect of tightened ramp up and ramp down constraints. We
have also identified two areas where a further improvement is needed.
Some very promising preliminary numerical results of the further improvements
are also included.

\bibliographystyle{siam}
\bibliography{mat_fin_bib,transfer}

\end{document}